\documentclass[reqno,12pt]{amsart}

\usepackage{amssymb, amsmath, amsthm}
\usepackage[colorlinks=true,linkcolor=blue,citecolor=red]{hyperref}
\usepackage{color}
\usepackage{epsfig}
\usepackage{amsmath}
\allowdisplaybreaks
\usepackage{amssymb}
\usepackage{amscd}
\usepackage{graphicx}
\usepackage{mathrsfs}
\usepackage{enumerate}
\usepackage{tikz}
\usepackage{float}
\usepackage{epstopdf}
\numberwithin{equation}{section}
\newtheorem{lemma}{Lemma}[section]
   \newtheorem{thm}{Theorem}[section]
   \newtheorem{prop}{Proposition}[section]
   \newtheorem{rem}[thm]{Remark}

\usepackage{epsfig}

\numberwithin{equation}{section}
\numberwithin{prop}{section}
\numberwithin{lemma}{section}
\numberwithin{re}{section}
\numberwithin{coro}{section}

\newcommand{\R}{\mathbb{R}}
\newcommand{\C}{\mathbb{C}}

\keywords{Integrable system,
Massive Thirring model, Inverse scattering transform,
 Riemann-Hilbert problem, Double-soliton}

\begin{document}

\title[Double-soliton solutions of MTM]{Exponential and algebraic double-soliton  solutions of the massive Thirring model}

\author[Z. Q. Li]{Zhi-Qiang Li}
\address{ Zhi-Qiang Li, \newline
School of Mathematics, China University of Mining and Technology, Xuzhou 221116, People's Republic of China and
Department of Mathematics and Statistics, McMaster University, Hamilton, Ontario, Canada, L8S 4K1}
\email{zqli@cumt.edu.cn, li3279@mcmaster.ca}

\author[D. E. Pelinovsky]{Dmitry E. Pelinovsky}
\address{ Dmitry E. Pelinovsky, \newline
Department of Mathematics and Statistics, McMaster University, Hamilton, Ontario, Canada, L8S 4K1}
\email{pelinod@mcmaster.ca}

\author[S. F. Tian]{Shou-Fu Tian}
\address{ Shou-Fu Tian,  \newline
School of Mathematics, China University of Mining and Technology, Xuzhou 221116, People's Republic of China}
\email{sftian@cumt.edu.cn}

\begin{abstract}
The newly discovered exponential and algebraic double-soliton solutions of the massive Thirring model in laboratory coordinates are placed in the context of the inverse scattering transform. We show that the exponential double-solitons correspond to double isolated eigenvalues in the Lax spectrum, whereas the algebraic double-solitons correspond to double embedded eigenvalues on the imaginary axis, where the continuous spectrum resides. This resolves the long-standing conjecture that multiple embedded eigenvalues may exist in the spectral problem associated with the massive Thirring model. To obtain the exponential double-solitons, we solve the Riemann--Hilbert problem with the reflectionless potential in the case of a quadruplet of double poles in each quadrant of the complex plane. To obtain the algebraic double-solitons, we consider the singular limit where the quadruplet of double poles degenerates into a symmetric pair of double embedded poles on the imaginary axis.   
\end{abstract}

\date{\today}
\maketitle


\section{Introduction}

We address the massive Thirring model (MTM) in laboratory coordinates, which  can be written in the following normalized form
\begin{align}\label{MTM}
\begin{split}
i(u_t+u_x)+v+|v|^2u=0,\\
i(v_t-v_x)+u+|u|^2v=0,
\end{split}
\end{align}
where $u=u(x,t)$ and $v=v(x,t)$ are complex functions of real variables $x$ and $t$. The MTM was introduced in \cite{Thirring-AOP-1958} in the context of quantum field theory as a relativistically invariant nonlinear Dirac equation in one spatial dimension. It was found in \cite{Mikhailov-JETP-1976} (see also \cite{Kaup-ANC-1977,Kuznetsov-TMP-1977,Orfanidis-PRD-1976}) that the MTM is a commutativity condition for a Lax pair of linear equations, hence it is completely integrable by the inverse scattering transform (IST) method. The Lax pair of linear equations for the MTM is given by
\begin{align}\label{Lax-pair}
\partial_x \psi =  L(u,v,\zeta) \psi, \quad \partial_t \psi = A(u,v,\zeta) \psi,
\end{align}
where $\zeta \in \mathbb{C}$ is the spectral parameter,
$\psi = \psi(x,t) \in \mathbb{C}^2$ is the wave function, and
the $2$-by-$2$ matrices $L(u,v,\zeta)$ and $A(u,v,\zeta)$ are given by
\begin{align*}
L=\frac{i}{4}\left(|u|^2-|v|^2\right)\sigma_3-\frac{i}{2}\zeta\left(
\begin{array}{cc}
0 & \bar{v} \\
v & 0 \\
\end{array}
\right)+\frac{i}{2\zeta}\left(
\begin{array}{cc}
0 & \bar{u} \\
u & 0 \\
\end{array}
\right)+\frac{i}{4}\left(\zeta^2-\zeta^{-2}\right)\sigma_3
\end{align*}
and
\begin{align*}
A=-\frac{i}{4}\left(|u|^2+|v|^2\right)\sigma_3-\frac{i}{2}\zeta\left(
\begin{array}{cc}
0 & \bar{v} \\
v & 0 \\
\end{array}
\right)-\frac{i}{2\zeta}\left(
\begin{array}{cc}
0 & \bar{u} \\
u & 0 \\
\end{array}
\right)+\frac{i}{4}\left(\zeta^2+\zeta^{-2}\right)\sigma_3.
\end{align*}
Here the bar stands for the complex conjugation and $\sigma_3 = {\rm diag}(1,-1)$ is the third Pauli's matrix. The compatibility condition $\partial_t \partial_x \psi = \partial_x \partial_t \psi$ in the linear system (\ref{Lax-pair}) coincides with the MTM system \eqref{MTM}.

The IST method based on the Riemann--Hilbert (RH) problem has been applied
for the Lax pair (\ref{Lax-pair}) in the recent works \cite{Pelinovsky-FIC-2019} and  \cite{He-MTM-slo-res-2023,LiLiXu-2024} (see also earlier works \cite{V} and \cite{Lee}). The IST method is used to obtain global solutions and to study the long-time dynamics of the MTM system \eqref{MTM}
for the initial-value problem with the initial data $(u,v)|_{t=0} = (u_0,v_0)$ decaying to zero at infinity. The decay condition
on $(u_0,v_0)$ is required to be sufficiently fast so that the functions and their first and second derivatives are square integrable with the weight
$\sqrt{1+x^2}$ \cite{He-MTM-slo-res-2023}. Exponential solitons satisfy this requirement and each soliton corresponds to a quadruplet of simple poles of the RH problem in each  quadrant of the complex plane, or equivalently to simple isolated eigenvalues in the Lax spectrum of the linear system (\ref{Lax-pair}). However, algebraic solitons decay as $(u,v) = \mathcal{O}(|x|^{-1})$ as $|x| \to \infty$ and hence they are not included in the IST method. Each algebraic soliton corresponds to a simple embedded eigenvalue
in the Lax spectrum located on the imaginary axis (no embedded eigenvalues exist on the real axis).

The algebraic solitons in the MTM were studied in \cite{Klaus-JNS-2006}, where the perturbation theory for embedded eigenvalues in the Lax spectrum
of the linear system (\ref{Lax-pair}) was developed. It was shown in \cite[Proposition 7.1]{Klaus-JNS-2006} that a pair of simple embedded eigenvalues on the imaginary axis is structurally unstable  and moves into a quadruplet of simple isolated eigenvalues in each quadrant of the complex plane under a generic perturbation of the initial data. A possibility of embedded eigenvalues of a higher algebraic multiplicity was also suggested in \cite[Lemma 6.4]{Klaus-JNS-2006} with some
precise conditions on the spatial decay of eigenvectors and generalized eigenvectors at infinity. Such embedded eigenvalues of higher algebraic multiplicity generally correspond to rational solutions of the MTM describing algebraic multi-solitons. However, the existence of such rational solutions has not been established in the literature up to very recently,
despite many works on rational solutions in integrable systems (see, e.g., \cite{Clarkson1,Clarkson2,Pelin1,Pelin2,Yang1,Yang2,Yang3}).

Rational solutions of the MTM were constructed on the constant nonzero background in \cite{Feng2,He,Ye}. They are relevant to dynamics of rogue waves on the modulationally unstable background but do not describe the dynamics of algebraic solitons at the zero background. It was only recently shown in \cite{Han-2024} (based on the Hirota's bilinear method developed in \cite{Chen-SAPM-2023}) that the algebraic double-solitons exist as the exact solutions of the MTM suggesting the existence of the higher-order algebraic solitons in a hierarchy of rational solutions to the MTM. Within the bilinear method, it was not shown in \cite{Han-2024} that the algebraic double-solitons correspond to the double embedded eigenvalues in the Lax spectrum predicted in \cite{Klaus-JNS-2006}.

{\em The main motivation for our work is to use the RH problem and to obtain the algebraic double-solitons of the MTM system (\ref{MTM}) associated with the double embedded eigenvalues in the Lax spectrum of the linear system (\ref{Lax-pair}).} To derive this result, we construct the exponential double-solitons associated with a quadruplet of double isolated eigenvalues in each quadrant of the complex plane and take the singular limit when the quadruplet of double isolated eigenvalues transforms into a symmetric pair of double embedded eigenvalues on the imaginary axis.

The study of double eigenvalues has started with the pioneering work  \cite{Zakharov-JETP-1972}, where it was shown that the double eigenvalues
of the associated spectral problem give the exponential double-solitons describing the slow (logarithmic in time) dynamics of two identical solitons of the focusing nonlinear Schr\"{o}dinger (NLS) equation. Properties of such exponential double-solitons were recently studied in nonintegrable versions of the NLS equation in \cite{Martel2020}. The exponential double-solitons on the nonzero constant background were constructed in \cite{Pichler-IMAJAM-2017}  after the development in the IST methods
on the nonzero background in \cite{Biondini-JMP-2014}.

The double-soliton solutions in the closely related
derivative NLS equation were constructed by using the Darboux transformations in \cite{Xu-JPA-2011,Zhang-CNSNS-2014} and \cite{Guo-SAPM-2013}. It was understood in \cite{Zhang-CNSNS-2014} that the algebraic double-solitons arise from the exponential double-solitons in the singular limit, for which the modified Darboux transformations have been developed in \cite{Guo-SAPM-2013}.
The IST method was also employed in the context of the derivative NLS equation to construct the exponential double-solitons from the double poles of the RH problem in \cite{Wang,Zhang-JNS-2020,Zhang-PD-2019,Zhang-TMP-2019}. Although both the derivative NLS equation and the MTM system in characteristic coordinates are related to the same spectral problem \cite{KN-1978,Kaup-ANC-1977}, the computational details for the MTM system in laboratory coordinates are different and technically more complicated. We close this gap in the literature by presenting the exponential double-solitons of the MTM system (\ref{MTM}) for the double isolated eigenvalues of the linear system (\ref{Lax-pair}).
{\em The main application of this result is to obtain the algebraic double-solitons and the double embedded eigenvalues in the singular limit, where the RH problem cannot be used. }

For the spectral problem associated with the focusing NLS equation on a nonzero background, it was understood in \cite{BilmanMiller2019} how to modify the RH problem for the simple and multiple embedded eigenvalues at the end points of the continuous spectrum in order to construct the rogue waves \cite{BilmanLingMiller2019}. This modification of the RH problem has not been developed so far for the spectral problem associated with the derivative NLS equation and the MTM system on the zero background. It is still unclear how the simple or multiple embedded eigenvalues can be constructed in the RH problem directly. We hope that our work will motivate further study of the associated spectral problems with embedded eigenvalues.

This paper is organized as follows. Section \ref{sec-2} introduces the RH problem for the MTM and formulates the main results.
The exponential double-solitons are constructed in Section \ref{sec-3} from the isolated double-pole solutions of the RH problem. The algebraic
double-solitons are obtained in Section \ref{sec-4}  by taking the singular limit to the embedded double-pole solutions of the RH problem. Appendix \ref{app-A} reports similar computations for the exponential and algebraic single-solitons for convenience of readers. Appendix \ref{app-B} reviews the construction of the exponential double-solitons in the MTM system by using the bilinear Hirota method.

\section{RH problem for MTM and main results}
\label{sec-2}

Assume that $(u,v) \to (0,0)$ as $|x| \to \infty$ fast enough, see Lemmas \ref{lemma-1} and \ref{lemma-2} below for precise requirements on $(u,v)$.
We define the matrix Jost functions for the linear system (\ref{Lax-pair}) from the boundary conditions:
\begin{align}
\label{asymptotic-Lax-pair}
  \psi^{(\pm)}(\zeta,x,t) \to \left( \begin{matrix} e^{\frac{i}{4}\left(\zeta^2-\zeta^{-2}\right) x + \frac{i}{4}\left(\zeta^2+\zeta^{-2}\right) t} & 0 \\
  0 & e^{-\frac{i}{4}\left(\zeta^2-\zeta^{-2}\right) x - \frac{i}{4}\left(\zeta^2+\zeta^{-2}\right) t} \end{matrix} \right) \quad \mbox{\rm as} \quad x \to \pm \infty.
\end{align}
For simplicity of notations, we will drop the dependence of $\psi^{(\pm)}$ on $(x,t)$. Since the Jost functions $\psi^{\pm}(\zeta)$ represent the fundamental matrix solutions of the linear sytem (\ref{Lax-pair}), they are related to each other by the scattering relations introduced for $\zeta \in (\R \cup i \R) \backslash \{0\}$ as
\begin{align}
\label{SMmatrix-1}
\psi^{(-)}(\zeta) = \psi^{(+)}(\zeta) \left(
\begin{array}{cc}
\overline{a}(\zeta) &   b(\zeta) \\
 -\overline{b}(\zeta)  & a(\zeta) \\
\end{array}
\right),
\end{align}
where the symmetry of scattering coefficients $a(\zeta)$ and $b(\zeta)$ follows from the symmetry of matrix Jost functions:
\begin{align}
\label{symmetry-psi}
	\psi^{(\pm)}(\zeta) = \left(
	\begin{array}{cc}
	0 & -1 \\
	1 & 0 \\
	\end{array}
	\right) \overline{\psi}^{(\pm)}(\zeta)
	\left(
	\begin{array}{cc}
	0 & 1 \\
	-1 & 0 \\
	\end{array}
	\right).
	\end{align}

As is explained in \cite{Pelinovsky-FIC-2019}, the linear system (\ref{Lax-pair}) can be folded to the squared
spectral parameter $\lambda := \zeta^2$ in two different ways,
one is suitable near $\zeta = 0$ and the other one is suitable
near $\zeta = \infty$. Following \cite{He-MTM-slo-res-2023}, we will
only consider the second transformation, from which we will define the
Riemann-Hilbert (RH) problem and solve it for the exponential
double-solitons, see Theorem \ref{explicit-expression-NU-Muv} below.

Hence we introduce the modified Jost functions as
\begin{align}
\label{Transformation-1}
\left\{ \begin{array}{l}
n_1^{(\pm)}(\lambda) := T(v,\zeta)
\psi_1^{(\pm)}(\zeta) e^{-\frac{i}{4}\left(\zeta^2-\zeta^{-2}\right) x - \frac{i}{4}\left(\zeta^2+\zeta^{-2}\right) t}, \\
n_2^{(\pm)}(\lambda) := \zeta^{-1} T(v,\zeta) \psi_2^{(\pm)}(\zeta) e^{\frac{i}{4}\left(\zeta^2-\zeta^{-2}\right) x + \frac{i}{4}\left(\zeta^2+\zeta^{-2}\right) t},\end{array}\right.
\end{align}
where the subscripts indicate the columns of the $2$-by-$2$ matrices
and the transformation matrix is given by
\begin{align*}
T(v,\zeta) := \left(
                       \begin{array}{cc}
                         1 & 0 \\
                         v & \zeta \\
                       \end{array}
                     \right).
\end{align*}
It follows from (\ref{asymptotic-Lax-pair}) that the modified Jost functions
satisfy
\begin{align*}
n_1^{(\pm)}(\lambda) \to e_1 := \left( \begin{matrix} 1\\0 \end{matrix} \right), \quad n_2^{(\pm)}(\lambda) \to e_2 := \left( \begin{matrix} 0\\1 \end{matrix} \right) \quad \mbox{\rm as} \quad x \to \pm \infty.
\end{align*}
Moreover, $n_{1,2}^{(\pm)}(\lambda)$ satisfy the integral equations, from
which the following properties were proven in \cite[Lemmas 3--5]{Pelinovsky-FIC-2019}.

\begin{lemma}
	\label{lemma-1}
Let $(u,v)\in L^1(\mathbb R)\cap L^{\infty}(\mathbb R)$ and $(u_x,v_x)\in L^1(\mathbb R)$. For every $\lambda \in \mathbb R\setminus \{0\}$, there exists unique bounded Jost functions $n_1^{(\pm)}(\lambda)$ and  $n_2^{(\pm)}(\lambda)$. For every $x \in\mathbb R$,  $n_1^{(\pm)}$ and  $n_2^{(\pm)}$ are continued analytically in $\mathbb{C}^{\pm}$ and
satisfy the following limits as $|\lambda| \to\infty$
and $\lambda \to 0$  along a contour in the domains of their analyticity:
\begin{align}
\label{boundary-cond-1}
 \lim_{|\lambda| \to \infty}\frac{n_1^{(\pm)}(\lambda)}{n^{\pm \infty}_{1}} = e_1, \quad
 \lim_{|\lambda| \to \infty}\frac{n_2^{(\pm)}(\lambda)}{n^{\pm \infty}_{2}} = e_2,
\end{align}
and
\begin{align}
\label{boundary-cond-2}
\lim_{\lambda \to 0} \left[n^{\pm \infty}_{1} n_1^{(\pm)}(\lambda) \right]= e_1 + v e_2, \quad
\lim_{\lambda \to 0} \left[n^{\pm \infty}_{2} n_2^{(\pm)}(\lambda) \right]= \bar{u} e_1 + (1 + \bar{u} v) e_2,
\end{align}
where
\begin{align*}
n^{\pm \infty}_{1} := e^{\frac{i}{4}\int_{\pm \infty}^{x}(|u|^2+|v|^2) dy}, \quad
n^{\pm \infty}_{2} := e^{-\frac{i}{4}\int_{\pm \infty}^{x}(|u|^2+|v|^2) dy}.
\end{align*}
\end{lemma}

Recall that $\lambda := \zeta^2$ and that $\lambda \in \R$ for $\zeta \in (\R\cup i\R)$. Hence we define new scattering coefficients for $\lambda \in \R \backslash \{0\}$ as
\begin{align*}
  \alpha(\lambda) := a(\zeta), \quad \beta_{+}(\lambda) :=\zeta b(\zeta),\quad  \beta_{-}(\lambda) := \zeta^{-1} b(\zeta).
\end{align*}
After the folding transformation (\ref{Transformation-1}),
the scattering relations (\ref{SMmatrix-1}) are modified as follows
\begin{align}\label{n12-S-n12-relation}
n^{(-)}(\lambda) = n^{(+)}(\lambda)  \left(
\begin{array}{cc}
 \overline{\alpha}(\lambda) &  \beta_{-}(\lambda) e^{2i\theta(\lambda)} \\
 -\overline{\beta}_{+}(\lambda) e^{-2i \theta(\lambda)}&  \alpha(\lambda) \\
\end{array}
\right),
\end{align}
where
\begin{align}
\label{theta}
 \theta(\lambda) := \frac{1}{4}(\lambda-\lambda^{-1})x+\frac{1}{4}(\lambda+\lambda^{-1})t.
\end{align}
The following lemma was proven in \cite[Lemma 6]{Pelinovsky-FIC-2019}.

\begin{lemma}
		\label{lemma-2}
Let $(u,v)\in L^1(\mathbb R)\cap L^{\infty}(\mathbb R)$ and $(u_x,v_x)\in L^1(\mathbb R)$. Then, $\alpha$ is continued analytically into $\mathbb C^{+}$  with the following limits in $\mathbb C^{+}$:
\begin{align}
\label{coefficient-1}
\lim_{|\lambda|\to \infty}\alpha(\lambda)=e^{-\frac{i}{4}\int_{\mathbb R}(|u|^2+|v|^2) dy}
\end{align}
and
\begin{align}
\label{coefficient-2}
\lim_{\lambda \to 0}\alpha(\lambda)=e^{\frac{i}{4}\int_{\mathbb R}(|u|^2+|v|^2) dy},
\end{align}
whereas $\beta_{\pm}$  are not continued analytically outside $\R$ and satisfy the limits
\begin{align*}
\lim_{|\lambda| \to \infty}\beta_{\pm}(\lambda) = \lim_{\lambda \to 0} \beta_{\pm}(\lambda) = 0.
\end{align*}
\end{lemma}

The RH problem for the modified Jost functions $n^{(\pm)}(\lambda)$ is constructed as follows.  We first define the sectionally meromorphic matrix $P(\lambda)\in \mathbb C^{2\times2}$ by
\begin{align}
\label{P-matrix}
P(\lambda) := \left\{\begin{aligned}
&\left( n^{(+)}_{1}(\lambda),\frac{ n^{(-)}_{2}(\lambda)}{ \alpha(\lambda)}\right), \quad \lambda \in \mathbb{C}^{+},\\
&\left(\frac{ n^{(-)}_{1}(\lambda)}{\bar{\alpha}(\lambda)}, n^{(+)}_{2}(\lambda)\right), \quad \lambda \in \mathbb{C}^{-}.
\end{aligned}\right.
\end{align}
By using (\ref{boundary-cond-1}) and (\ref{coefficient-1}),
we obtain the following limits as $|\lambda| \to \infty$ in the domain
of meromorphicity of $P(\lambda)$:
\begin{align}\label{p-infty}
\lim_{|\lambda| \to \infty} P(\lambda) = \left(
                           \begin{array}{cc}
                              n^{+\infty}_{1} & 0 \\
                             0 &   n^{+\infty}_{2} \\
                           \end{array}
                         \right) =: P^{\infty},
\end{align}
where $\overline{P^{\infty}} = (P^{\infty})^{-1}$.
We finally define $M(\lambda) := (P^{\infty})^{-1} P(\lambda)$ and
formulate the normalized RH problem.

\vspace{0.2cm}
\centerline{\fbox{\parbox[cs]{1.0\textwidth}{
{\bf RH problem.} Find a complex-valued function $M(\lambda)$ with the following properties:
\begin{itemize}
  \item $ M(\lambda)$ is meromorphic in $\mathbb{C}\setminus\mathbb{R}$.
  \item $ M(\lambda) \to \mathbb I$ as $|\lambda| \to \infty$, where $\mathbb I$ is the $2$-by-$2$ identity matrix.
  \item $ M_{+}(\lambda)= M_{-}(\lambda) V(\lambda)$ for every $\lambda\in\mathbb{R}$, where $M_{\pm}(\lambda) := \lim\limits_{{\rm Im}(\lambda) \to \pm 0} M(\lambda)$ and
  \begin{align*}
 V(\lambda) := \left(\begin{array}{cc}
                   1 & - r_-(\lambda)e^{2i\theta(\lambda)} \\
                   -\overline{r}_+(\lambda) e^{-2i\theta(\lambda)} & 1 +  \overline{r}_+(\lambda) r_-(\lambda)
                 \end{array}\right), \quad
                 r_{\pm}(\lambda) := \frac{\beta_{\pm}(\lambda)}{\alpha(\lambda)}.
\end{align*}
\end{itemize}
}}}

It follows from (\ref{boundary-cond-2}) and (\ref{coefficient-2}) (see also \cite[Proposition 2.24]{He-MTM-slo-res-2023}) that
the potentials $(u,v)$ for solutions of the MTM system (\ref{MTM}) can be recovered from solutions of the RH problem by using the following asymptotic limits taken in the domains of meromorphicity of $M(\lambda)$:
\begin{align}
u = \lim_{\lambda \to 0} \overline{M}_{12}(\lambda), \quad
v = \lim_{\lambda \to 0}  M_{21}(\lambda).
\label{recover-potential}
\end{align}

Solvability of the RH problem under some conditions of the reflection
coefficients $r_{\pm}(\lambda)$ was studied in \cite{He-MTM-slo-res-2023,Pelinovsky-FIC-2019}. In this work,
we consider the reflectionless case $r_{\pm}(\lambda) \equiv 0$ for $\lambda \in \R$ in the particular case when $\alpha(\lambda)$ admits a double pole
at $\lambda_0 \in \mathbb{C}^+$.

It is well-known (see, e.g., \cite{He-MTM-slo-res-2023,LiLiXu-2024}) that a simple pole of $\alpha(\lambda)$ leads to a single-soliton solution. For completeness, we give details of the RH problem with a simple pole in Appendix \ref{app-A}. To simplify the presentation of soliton solutions, we should use the basic symmetries of the MTM system. In particular, the relativistically invariant MTM system (\ref{MTM}) admits the Lorentz symmetry
\begin{equation}
\label{MTM-Lorentz}
\left[ \begin{matrix}
u(x,t) \\
v(x,t)
\end{matrix} \right]  \;\; \mapsto \;\;
\left[ \begin{matrix}
\left( \frac{1 - c}{1 + c} \right)^{1/4} u\left( \frac{x + ct}{\sqrt{1-c^2}}, \frac{t + cx}{\sqrt{1-c^2}} \right) \\
\left( \frac{1 + c}{1 - c} \right)^{1/4} v\left( \frac{x + ct}{\sqrt{1-c^2}}, \frac{t + cx}{\sqrt{1-c^2}} \right)
\end{matrix}\right],
\qquad c \in (-1,1).
\end{equation}
In addition, it admits the translational and rotational symmetries
\begin{equation}
\label{MTM-symm}
\left[ \begin{matrix}
u(x,t) \\
v(x,t)
\end{matrix} \right]  \quad \mapsto \quad
\left[ \begin{matrix}
u(x+x_0,t+t_0) e^{\mathrm{i} \theta_0} \\
v(x+x_0,t+t_0) e^{\mathrm{i} \theta_0}
\end{matrix}\right],
\qquad x_0, t_0, \theta_0 \in \mathbb{R}.
\end{equation}
By using (\ref{MTM-Lorentz}) and (\ref{MTM-symm}), the single-soliton solutions can be expressed in a short form:
\begin{equation}
\label{one_soliton_1}
\left\{
\begin{array}{lr}
u(x,t) = i (\sin \gamma) {\rm sech}\left( x \sin \gamma - i \frac{\gamma}{2}\right) e^{-i t \cos \gamma},\\
v(x,t)=-i (\sin \gamma) {\rm sech}\left( x \sin \gamma + i \frac{\gamma}{2}\right) e^{-i t \cos \gamma},
\end{array}
\right.
\end{equation}
where $\gamma\in(0,\pi)$ is a free parameter. More general single-soliton solutions can be extended with speed parameter $c \in (-1,1)$ by using (\ref{MTM-Lorentz}) and with two translational parameters $x_0,t_0 \in \R$ by using (\ref{MTM-symm}), where translation in $\theta_0$ is linearly dependent from translation in $t_0$.

The normalized single-soliton solution (\ref{one_soliton_1}) corresponds to the simple pole of the RH problem at $\lambda_0 = e^{i \gamma} \in \mathbb{C}^+$ with $\gamma \in (0,\pi)$, see Appendix \ref{app-A}. The double-soliton solutions will also be constructed for $\lambda_0 = e^{i \gamma} \in \mathbb{C}^+$. The following theorem gives the explicit representation of the double-soliton solutions.
As we show in Appendix \ref{app-B}, this representation coincides with the explicit formula obtained by the bilinear Hirota method developed in \cite{Chen-SAPM-2023}.

\begin{thm}\label{explicit-expression-NU-Muv}
	Let $\lambda_0 = e^{i \gamma}$ with $\gamma \in (0,\pi)$ be a double pole of the RH problem.  Then, the solution $(u,v)$ of the MTM system (\ref{MTM}) obtained from (\ref{recover-potential}) is given by
	\begin{align}
	\label{double-soliton}
u = \frac{\bar{N}_u}{\overline{D(M)}}, \quad
v = \frac{N_v}{D(M)},
	\end{align}
	where
\begin{align*}
N_u &= 4i (\sin \gamma)^2 e^{-x \sin \gamma + i t  \cos \gamma - \frac{i \gamma}{2}} \left( (x-\tilde{x}_0) \cos \gamma + i (t-\tilde{t}_0) \sin \gamma + i \right. \\
& \left. \quad
- e^{-2 x \sin \gamma - i \gamma} [ 2 \cot \gamma + (x-\tilde{x}_0) \cos \gamma - i (t-\tilde{t}_0) \sin \gamma] \right),
\end{align*}
\begin{align*}
N_v &= 4i (\sin \gamma)^2 e^{-x \sin \gamma - i t  \cos \gamma - \frac{i \gamma}{2}} \left( (x-\tilde{x}_0)  \cos \gamma - i (t-\tilde{t}_0) \sin \gamma \right. \\
& \left. \quad
- e^{-2 x \sin \gamma - i \gamma} [ 2 \cot \gamma + (x-\tilde{x}_0)  \cos \gamma + i (t-\tilde{t}_0) \sin \gamma + i] \right),
\end{align*}
	and
\begin{align*}
D(M) &= 1 + e^{-4 x \sin \gamma - 2 i \gamma} + 2 e^{-2 x \sin \gamma - i \gamma} \\
& \quad \times
\left( 1 + 2 (\sin \gamma)^2 \left[  \cot \gamma + (x-\tilde{x}_0) \cos \gamma + \frac{i}{2} \right]^2 + 2 (\sin \gamma)^4 \left[t-\tilde{t}_0 + \frac{1}{2 \sin \gamma} \right]^2 \right),
\end{align*}
where $\tilde x_0,\tilde t_0 \in \mathbb{R}$ are arbitrary parameters
in addition to parameters $c \in (-1,1)$ and $x_0,t_0 \in \R$
obtained from the transformations (\ref{MTM-Lorentz}) and (\ref{MTM-symm}).
\end{thm}

\begin{rem}
Parameter $\tilde{t}_0$ is trivially removed by using translational symmetries (\ref{MTM-symm}) with translations in $\theta_0$ and $t_0$. Hence, the double-soliton solutions of Theorem \ref{explicit-expression-NU-Muv} only have two non-trivial parameters: $\gamma \in (0,\pi)$ and $\tilde{x}_0 \in \R$. 	
\end{rem}

Although the explicit form of double-soliton solutions in Theorem \ref{explicit-expression-NU-Muv} can be obtained by algebraic methods such as Darboux transformations or the bilinear Hirota method, see Appendix \ref{app-B}, the RH problem enables us to clarify the Lax spectrum of the double-soliton solutions. Based on the solution in Section \ref{sec-3}, we prove that
$\zeta_0 := \sqrt{\lambda_0} = e^{\frac{i}{2} \gamma}$ is a double eigenvalue
of the linear system (\ref{Lax-pair}) with only one eigenvector $\psi_0 \in H^1(\R,\C^2)$ and one generalized eigenvector $\psi_1 \in H^1(\R,\C^2)$. The eigenvector and the generalized eigenvector satisfy the following linear equations:
\begin{align}
\label{eigenvector}
\partial_x \psi_0 =  L(u,v,\zeta_0) \psi_0, \quad \partial_t \psi_0 = A(u,v,\zeta_0) \psi_0
\end{align}
and
\begin{align}
\label{eigenvector-generalized}
\partial_x \psi_1 =  L(u,v,\zeta_0) \psi_1 + \partial_{\zeta} L(u,v,\zeta_0) \psi_0, \quad \partial_t \psi_1 = A(u,v,\zeta_0) \psi_1 + \partial_{\zeta} A(u,v,\zeta_0) \psi_0,
\end{align}
where $(u,v)$ is given by (\ref{double-soliton}) and $\zeta_0 = e^{\frac{i \gamma}{2}}$.

The knowledge of eigenvectors and generalized eigenvectors in (\ref{eigenvector}) and (\ref{eigenvector-generalized}) is particularly important when the exponential double-soliton solution of Theorem \ref{explicit-expression-NU-Muv} converges as $\gamma \to \pi$ to the algebraic double-soliton solution obtained in \cite{Han-2024}. The following theorem states that the corresponding Lax spectrum includes the double embedded eigenvalue $\zeta_0 = i$ of the linear system (\ref{Lax-pair}) with only one eigenvector $\psi_0 \in H^1(\R,\C^2)$ and one generalized eigenvector $\psi_1 \in H^1(\R,\C^2)$ satisfying (\ref{eigenvector}) and (\ref{eigenvector-generalized}) for $\zeta_0 = i$.

\begin{thm}\label{Theorem-algebra-eigen-general}
		Let $\lambda_0 = e^{i \gamma}$ be a double pole of the RH problem with the solution $(u,v)$ of the MTM system (\ref{MTM}) obtained in Theorem \ref{explicit-expression-NU-Muv}. With proper choice of
		$\tilde x_0$ and $\tilde{t}_0$, this solution transforms in the limit $\gamma \to \pi$ to the form:
		\begin{align}\label{u-algebric-soliton}
		u_{\rm alg}(x,t) = \frac{-\frac{8}{3} x^3 - 4 i x^2 + 2x - i - 4 i (t - \tilde{t}_0) (i + 2 x) + 8 \tilde{x}_0} {\frac{4}{3}x^4 +\frac{8}{3} ix^3 + 2x^2 + 2 i x -\frac{1}{4} - 4 (t - \tilde{t}_0)^2 + 4 \tilde{x}_0 (i + 2x)} e^{it}
		\end{align}
		and
		\begin{align}\label{v-algebric-soliton}
	v_{\rm alg}(x,t) = \frac{-\frac{8}{3} x^3 + 4 i x^2 + 2x + i + 4 i (t - \tilde{t}_0) (i - 2 x) + 8 \tilde{x}_0}{\frac{4}{3}x^4 - \frac{8}{3} ix^3 + 2x^2 - 2 i x -\frac{1}{4} - 4 (t - \tilde{t}_0)^2 - 4 \tilde{x}_0 (i - 2x)} e^{it},
	\end{align}
	where $\tilde x_0,\tilde t_0 \in \mathbb{R}$ are (new) arbitrary parameters
	in addition to parameters $c \in (-1,1)$ and $x_0,t_0 \in \R$
	obtained from the transformations (\ref{MTM-Lorentz}) and (\ref{MTM-symm}).
The linear equations (\ref{eigenvector}) and (\ref{eigenvector-generalized}) with $(u,v) = (u_{\rm alg},v_{\rm alg})$ and $\zeta_0 = i$ admit the eigenvector
 		\begin{align}\label{calim-eigenvector-expression}
\psi_0 = e^{-\frac{i}{2} t} T^{-1} \mathfrak{n}_0,
	\end{align}
	where $T^{-1} := \left[ T(v_{\rm alg},i) \right]^{-1} = \left(
	\begin{array}{cc}
		1 & 0 \\
		i v_{alg} & -i \\
	\end{array}
	\right)$ and $\mathfrak{n}_0 = (\mathfrak{n}_{01},\mathfrak{n}_{02})^T$ is given by
	\begin{align*}
\mathfrak{n}_{01} &= e^{\frac{i}{4}\int_{+\infty}^{x}(|u|^2+|v|^2) dy}\frac{2x^2-2i (t - \tilde{t}_0) +\frac{1}{2}} {\frac{4}{3} x^4 -\frac{8}{3}  i x^3  + 2 x^2   - 2 i x -\frac{1}{4} - 4(t - \tilde{t}_0)^2 - 4 \tilde{x}_0 (i - 2 x)},\\
\mathfrak{n}_{02} &= e^{-\frac{i}{4}\int_{+\infty}^{x}(|u|^2+|v|^2) dy + i t} \frac{-2ix^2- 4 x+2 (t - \tilde{t}_0) +\frac{3}{2}i} {\frac{4}{3} x^4 -\frac{8}{3}  i x^3  + 2 x^2   - 2 i x -\frac{1}{4} - 4(t - \tilde{t}_0)^2 - 4 \tilde{x}_0 (i - 2 x)}e^{it}.
	\end{align*}
 	and the generalized eigenvector
 \begin{align}\label{calim-generalized-eigenvector-expression}
 \psi_1 = 2 i e^{-\frac{i}{2} t} T^{-1} \mathfrak{n}_1+i\partial_{\zeta} T(v,\zeta) \psi_0,
 \end{align}
	where $\partial_{\zeta} T(v,\zeta)=\left(
                                          \begin{array}{cc}
                                            0 & 0 \\
                                            0 & 1 \\
                                          \end{array}
                                        \right)
$ and $\mathfrak{n}_1 = (\mathfrak{n}_{11},\mathfrak{n}_{12})^T$ is given by
\begin{align*}
\mathfrak{n}_{11} = &e^{\frac{i}{4}\int_{+\infty}^{x}(|u|^2+|v|^2) dy}
\frac{ -\frac{1}{3}i x^3-2x^2 - \frac{4}{3}i x+ x (t-\tilde{t}_0) +2i (t-\tilde{t}_0) -\frac{1}{2}- 2i \tilde x_0} {\frac{4}{3} x^4 -\frac{8}{3}  i x^3  + 2 x^2   - 2 i x -\frac{1}{4} - 4(t - \tilde{t}_0)^2 - 4 \tilde{x}_0 (i - 2 x)}, \\
\mathfrak{n}_{12} = &e^{-\frac{i}{4}\int_{+\infty}^{x}(|u|^2+|v|^2) dy + it}\frac{- \frac{1}{3}x^3-ix^2- \frac{15}{4}x +ix (t-\tilde{t}_0) + \frac{ 5}{4}i +3 (t-\tilde{t}_0)  - 2 \tilde x_0 } {\frac{4}{3} x^4 -\frac{8}{3}  i x^3  + 2 x^2   - 2 i x -\frac{1}{4} - 4(t - \tilde{t}_0)^2 - 4 \tilde{x}_0 (i - 2 x)}e^{it}.
	\end{align*}
	\end{thm}

\begin{rem}
The eigenvector $\psi_0$ and generalized eigenvector $\psi_1$ in (\ref{calim-eigenvector-expression}) and (\ref{calim-generalized-eigenvector-expression})
for the double embedded eigenvalue $\zeta = i$ satisfy the criterion for the spatial decay in \cite[Lemma 6.4]{Klaus-JNS-2006}, namely
$\psi_0 = \mathcal{O}(|x|^{-2})$ and $\psi_1 = \mathcal{O}(|x|^{-1})$ as $|x| \to \infty$.
\end{rem}

\begin{rem}
	The algebraic double-soliton given by  (\ref{u-algebric-soliton}) and (\ref{v-algebric-soliton}) reduces
to the explicit expression obtained in \cite{Han-2024} by using the transformation
\begin{align*}
x\to -x,\quad  t\to-t, \quad  u\to u,\quad  v\to-v,
\end{align*}
due to a different normalization of the MTM system used in \cite{Han-2024}.
\end{rem}

\section{Exponential double-solitons for a double pole}
\label{sec-3}

Here we study solutions of the normalized RH problem for the refelectionless potential $r_{\pm}(\lambda) \equiv 0$ for $\lambda \in \R$ with a double pole of $M(\lambda)$ at $\lambda_0 \in \mathbb{C}^+$. By symmetry (\ref{symmetry-psi}),
$\bar{\lambda}_0 \in \mathbb{C}^-$ is also a double pole of $M(\lambda)$. The normalized RH problem can be rewritten in the form:

\vspace{0.2cm}

\centerline{\fbox{\parbox[cs]{1.0\textwidth}{
			{\bf RH problem.} Find a complex-valued function $M(\lambda)$ with the following properties:
			\begin{itemize}
			\item $ M(\lambda)$ has double poles at $\lambda_0 \in \mathbb{C}^+$ and $\bar{\lambda}_0 \in \mathbb{C}^-$.
			\item $ M(\lambda) \to \mathbb I$ as $|\lambda| \to \infty$,
			where $\mathbb I$ is the $2$-by-$2$ identity matrix.
			\item $ M_{+}(\lambda)= M_{-}(\lambda) $ for every $\lambda\in\mathbb{R}$, where $M_{\pm}(\lambda) := \lim\limits_{{\rm Im}(\lambda) \to \pm 0} M(\lambda)$.
			\end{itemize}
}}}

\vspace{0.2cm}

In order to  regularize the RH problem, we   subtract the residue conditions in both sides of the formula $ M_{+}(\lambda)= M_{-}(\lambda)$ and obtain the following solution of the normalized RH problem:
\begin{align}\label{RH-solution-double}
    M(\lambda)=\mathbb I+\frac{\mathrm{Res}_{\lambda=\lambda_0}  M_{+}(\lambda)}{\lambda-\lambda_0} +\frac{\mathrm{Res}_{\lambda=\bar\lambda_0} M_{-}(\lambda)}{\lambda-\bar\lambda_0}
  +\frac{\mathrm{P}^{-2}_{\lambda=\lambda_0} M_{+}(\lambda)}{(\lambda-\lambda_0)^2} +\frac{\mathrm{P}^{-2}_{\lambda=\bar\lambda_0}  M_{-}(\lambda)}{(\lambda-\bar\lambda_0)^2},
\end{align}
where $\mathrm{Res}_{\lambda=\lambda_0}$ is the residue coefficient and
$\mathrm{P}^{-2}_{\lambda=\lambda_0}$ is the double pole coefficient at $\lambda = \lambda_0$.

\subsection{Computations of the residue coefficients}

In order to compute the residue coefficients, we use the following result.

\begin{lemma}\label{residue-condition-theorem}
Assume $f$ and $g$ be analytic in a complex region $\Omega\in\mathbb{C}$
such that $g$ has a double pole at $z_{0}\in\Omega$ with $g(z_{0})=g'(z_{0})=0$, $g''(z_{0})\neq0$, and $f(z_{0})\neq0$. The residue coefficients of $f/g$
at $z = z_0$ are given by
\begin{align*}
\mathrm{Res}_{z=z_{0}} \frac{f(z)}{g(z)} = \frac{2f'(z_{0})}{g''(z_{0})} -
\frac{2f(z_{0})g'''(z_{0})}{3 [g''(z_{0})]^{2}},\quad
\mathrm{P}^{-2}_{z=z_{0}} \frac{f(z)}{g(z)} = \frac{2f(z_{0})}{g''(z_{0})}.
\end{align*}
\end{lemma}

\begin{proof}
	Under conditions of the lemma, we have
	\begin{align*}
	f(z) &= f(z_0) + f'(z_0) (z-z_0) + \mathcal{O}((z-z_0)^2), \\
	g(z) &= \frac{1}{2} g''(z_0) (z-z_0)^2 + \frac{1}{6} g'''(z_0) (z-z_0)^3 + \mathcal{O}((z-z_0)^4),
	\end{align*}
from which the result follows by the Laurent expansion of $f(z)/g(z)$.
\end{proof}

The residue coefficients in (\ref{RH-solution-double}) are obtained
from (\ref{P-matrix}) and (\ref{p-infty}):
\begin{align*}
\mathrm{Res}_{\lambda=\lambda_0}  M_{+}(\lambda)= (P^{\infty})^{-1} \left( \vec{\textbf{0}} \quad \mathrm{Res}_{\lambda=\lambda_0} \frac{  n^{(-)}_{2}(\lambda)}{\tilde \alpha(\lambda)}\right),  &\quad \mathrm{P}^{-2}_{\lambda=\lambda_0}  M_{+}(\lambda) = (P^{\infty})^{-1} \left(\vec{\textbf{0}}\quad \mathrm{P}^{-2}_{\lambda=\lambda_0}\frac{  n^{(-)}_{2}(\lambda)}{ \alpha(\lambda)}\right), \\
\mathrm{Res}_{\lambda=\bar\lambda_0} \tilde M_{-}(\lambda) = (P^{\infty})^{-1} \left(\mathrm{Res}_{\lambda=\bar\lambda_0} \frac{  n^{(-)}_{1}(\lambda)}{
	\bar{\alpha}(\lambda)} \quad  \vec{\textbf{0}}\right), &\quad \mathrm{P}^{-2}_{\lambda=\bar\lambda_0}  M_{-}(\lambda) = (P^{\infty})^{-1} \left(\mathrm{P}^{-2}_{\lambda=\bar\lambda_0} \frac{ n^{(-)}_{1}(\lambda)}{\check{ \alpha }(\lambda)}\quad \vec{\textbf{0}}\right),
\end{align*}
where $\vec{\textbf{0}}$ is a $2$-by-$1$ zero vector. Based on Lemma \ref{residue-condition-theorem} and this representation, we obtain the residue coefficients in the following proposition.

\begin{prop}
The residue coefficients of $M_+(\lambda)$ at $\lambda = \lambda_0$ are given by
\begin{align}
\mathrm{P}^{-2}_{\lambda=\lambda_{0}}\left[\frac{ n^{(-)}_{2}(\lambda)}{ \alpha(\lambda)}\right]
&=  A_0 n^{(+)}_{1}(\lambda_0)e^{2i\theta(\lambda_0)},\label{P-2-lambda-0-1}\\
\mathrm{Res}_{\lambda=\lambda_{0}}\left[\frac{ n^{(-)}_{2}(\lambda)}{ \alpha (\lambda)}\right]
&= A_0e^{2i\theta(\lambda_0)}\left[ (n^{(+)}_{1})'(\lambda_0) + n^{(+)}_{1}(\lambda_0)\left(2i  \theta'(\lambda_0) +B_0\right)\right],\label{Res-lambda-0-1}
\end{align}
where $A_0$ and $B_0$ are arbitrary coefficients.
The residue conditions of $M_-(\lambda)$ at $\lambda = \bar\lambda_0$ are given by
\begin{align}
\mathrm{P}^{-2}_{\lambda=\bar\lambda_{0}} \left[\frac{ n^{(-)}_{1}(\lambda)}{\check{ \alpha }(\lambda)}\right]
&= - \bar{A}_0 \bar\lambda_0   n^{(+)}_{2}(\bar\lambda_0)e^{-2i\theta(\bar\lambda_0)},\label{P-2-bar-lambda-0-1}\\
\mathrm{Res}_{\lambda=\bar\lambda_{0}} \left[\frac{  n^{(-)}_{1}(\lambda)}{\bar{ \alpha }(\bar\lambda)}\right]
&= -\bar{A}_0 \bar\lambda_0e^{-2i\theta(\bar\lambda_0)}\left[   (n^{(+)}_{2})'(\bar\lambda_0) +  n^{(+)}_{2}(\bar\lambda_0)\left(-2i  \theta'(\bar\lambda_0) + \bar{B}_0 +\bar\lambda_0^{-1} \right)\right].
\label{Res-bar-lambda-0-1}
\end{align}
\end{prop}

\begin{proof}
	By assumption, $\lambda_0 = \zeta_0^2$ is a double zero of $\alpha(\lambda)$ extended to $\mathbb{C}^+$ by Lemma \ref{lemma-2}. Since it folows from \eqref{SMmatrix-1} that
	$$
	\alpha(\lambda) = a(\zeta) = \det(\psi_1^{(+)}(\zeta),\psi_2^{(-)}(\zeta)),
	$$
	we conclude that there exists a constant $e_0$ such that
\begin{align}\label{double-psi-2-psi-1}
  \psi^{(-)}_{2}(\zeta_0)=e_0\psi^{(+)}_{1}(\zeta_0).
\end{align}
Furthermore, since $\zeta_0$ is the double zero of $a(\zeta)$, we have
\begin{align*}
\begin{split}
  0=a'(\zeta_0)&=\det\left((\psi^{(+)}_{1})',\psi^{(-)}_{2}\right) +\det\left(\psi^{(+)}_{1},(\psi^{(-)}_{2})'\right)\big|_{\zeta=\zeta_0}\\
  &=\det\left(\psi^{(+)}_{1},-e_0(\psi^{(+)}_{1})'+(\psi^{(-)}_{2})'\right) \big|_{\zeta=\zeta_0},
\end{split}
\end{align*}
so that there exists another constant $h_0$ such that
\begin{align}\label{psi-2-psi+1+psi+1}
(\psi^{(-)}_{2})'(\zeta_0) = e_0 (\psi^{(+)}_{1})'(\zeta_0) + h_0 \psi^{(+)}_{1}(\zeta_0).
\end{align}

By using transformation \eqref{Transformation-1}, we rewrite
\eqref{double-psi-2-psi-1} as
\begin{align}\label{double-n2-n1}
  n^{(-)}_{2}(\lambda_0) = e_0 \zeta_0^{-1} n^{(+)}_{1}(\lambda_0)e^{2i\theta(\lambda_0)},
\end{align}
where $\theta(\lambda)$ is given by (\ref{theta}). This expression agrees with (\ref{n12-S-n12-relation}) for $\alpha(\lambda_0) = 0$. By using transformation \eqref{Transformation-1} again and the product rule, we derive
\begin{align*}
(n_1^{(+)})'(\lambda_0) &= (2\zeta_0)^{-1} T(v,\zeta_0) [(\psi_1^{(+)})'(\zeta_0) - 2 i \zeta_0 \theta'(\lambda) \psi_1^{(+)}(\zeta_0)] e^{-i \theta(\lambda_0)} \\
& \qquad + (2\zeta_0)^{-1} [\partial_{\zeta} T(v,\zeta_0)] \psi_1^{(+)}(\zeta_0) e^{-i \theta(\lambda_0)}, \\
(n_2^{(-)})'(\lambda_0) &= (2\zeta_0^2)^{-1} T(v,\zeta_0) [ (\psi_2^{(-)})'(\zeta) + 2 i \zeta_0 \theta'(\lambda_0) ] e^{i \theta(\lambda_0)} - (2 \zeta_0^2)^{-1} n_2^{(-)}(\lambda_0) \\
& \qquad + (2\zeta_0^2)^{-1} [\partial_{\zeta} T(v,\zeta_0)] \psi_2^{(-)}(\zeta_0) e^{i \theta(\lambda_0)},
\end{align*}
which imply due to (\ref{psi-2-psi+1+psi+1}) and (\ref{double-n2-n1}) that
\begin{align}\label{n-2-derivetive}
(n^{(-)}_{2})'(\lambda_0) = e^{2i\theta(\lambda_0)} \left[
e_0 \zeta_0^{-1} (n^{(+)}_{1})'(\lambda_0) + (2 \zeta_0^2)^{-1}
(h_0 + 4 i e_0 \zeta_0  \theta'(\lambda_0) -  e_0 \zeta_0^{-1})  n^{(+)}_{1}(\lambda_0) \right],
\end{align}
in agreement with the derivative of (\ref{n12-S-n12-relation}) at $\lambda = \lambda_0$.

We use the chain rule
\begin{align*}
\alpha'(\lambda) &= (2 \zeta)^{-1} a'(\zeta), \\
\alpha''(\lambda) &= (2 \zeta)^{-2} [a''(\zeta) - \zeta^{-1} a'(\zeta)], \\
\alpha'''(\lambda) &= (2 \zeta)^{-3} [a'''(\zeta) - 3 \zeta^{-1} a''(\zeta)
+ 3 \zeta^{-2} a'(\zeta)].
\end{align*}
By using \eqref{double-n2-n1} and \eqref{n-2-derivetive}, we compute
from the expressions in Lemma \ref{residue-condition-theorem} that
\begin{align}\label{P-2-lambda-0}
\mathrm{P}^{-2}_{\lambda=\lambda_{0}} \left[\frac{ n^{(-)}_{2}(\lambda)}{ \alpha (\lambda)}\right] = \frac{8 \zeta_0^2 n^{(-)}_{2}(\lambda_0)}{a''(\zeta_0)}
=\frac{8e_0\zeta_0}{a''(\zeta_0)}\ n^{(+)}_{1}(\lambda_0)e^{2i\theta(\lambda_0)}
\end{align}
and
\begin{align}\label{Res-lambda-0}
\begin{split}
\mathrm{Res}_{\lambda=\lambda_{0}}&\left[\frac{ n^{(-)}_{2}(\lambda)}{ \alpha (\lambda)}\right] =
\frac{8 \zeta_0^2 (n^{(-)}_{2})'(\lambda_0)}{a''(\zeta_0)} -
\frac{4 \zeta_0  n^{(-)}_{2}(\lambda_0) \left[ a'''(\zeta_0) - 3 \zeta_0^{-1} a''(\zeta_0) \right]}{3 [a''(\zeta_0)]^{2}}, \\
&= \frac{8 e_0 \zeta_0}{a''(\zeta_0)} e^{2i\theta(\lambda_0)}
\left[ (n^{(+)}_{1})'(\lambda_0) + n^{(+)}_{1}(\lambda_0)
\left( 2 i \theta'(\lambda_0) + \frac{h_0}{2 e_0 \zeta_0} -\frac{a'''(\zeta_0)}{6 \zeta_0 a''(\zeta_0)} \right) \right].
\end{split}
\end{align}
Let \begin{align}\label{A-0-B-0-definition}
A_0=\frac{8e_0\zeta_0}{a''(\zeta_0)}, \qquad B_0 = \frac{h_0}{2 e_0 \zeta_0} -\frac{a'''(\zeta_0)}{6 \zeta_0 a''(\zeta_0)},
\end{align}
then \eqref{P-2-lambda-0} and \eqref{Res-lambda-0} are transformed into \eqref{P-2-lambda-0-1} and \eqref{Res-lambda-0-1}.

By using the symmetry condition (\ref{symmetry-psi}), we have
\begin{align*}
\psi_1^{(\pm)}(\zeta)= \left(
\begin{array}{cc}
0 & 1 \\
-1 & 0 \\
\end{array}
\right) \overline{\psi}_2^{(\pm)}(\zeta),\quad
\psi_2^{(\pm)}(\zeta)= \left(
\begin{array}{cc}
0 & -1 \\
1 & 0 \\
\end{array}
\right) \overline{\psi}_1^{(\pm)}(\zeta),
\end{align*}
from which we obtain with the help of (\ref{double-psi-2-psi-1}) and (\ref{psi-2-psi+1+psi+1}) that
\begin{align*}
  \psi^{(-)}_{1}(\bar\zeta_0) = -\bar{e}_0 \psi^{(+)}_{2}(\bar\zeta_0)
\end{align*}
and
\begin{align*}
(\psi^{(-)}_{1})'(\bar\zeta_0) = -\bar{e}_0 (\psi^{(+)}_{2})'(\bar\zeta_0) -\bar{h}_0 \psi^{(+)}_{2}(\bar\zeta_0).
\end{align*}
Furthermore, by using the transformation \eqref{Transformation-1} and its derivative, similarly to (\ref{double-n2-n1}) and (\ref{n-2-derivetive}),
we obtain
\begin{align*}
n^{(-)}_{1}(\bar\lambda_0) &=  -\bar{e}_0 \bar{\zeta}_0 n^{(+)}_{2}(\bar\lambda_0)e^{-2i\theta(\bar\lambda_0)}
\end{align*}
and
\begin{align*}
(n^{(-)}_{1})'(\bar\lambda_0) &= - e^{-2i\theta(\bar\lambda_0)} \left[ \bar{e}_0 \bar{\zeta}_0 (n^{(+)}_{2})'(\bar\lambda_0)
+ \frac{1}{2} (\bar{h}_0 - 4 i \bar{e}_0 \bar{\zeta}_0 \theta'(\bar{\lambda}_0) + \bar{e}_0 \bar{\zeta}_0^{-1}) n^{(+)}_{2}(\bar\lambda_0)  \right].
\end{align*}
By using these expressions we compute from Lemma \ref{residue-condition-theorem}, similarly to (\ref{P-2-lambda-0}) and (\ref{Res-lambda-0}), that
\begin{align}
\label{P-2-barlambda-0}
\mathrm{P}^{-2}_{\lambda=\bar \lambda_{0}} \left[\frac{ n^{(-)}_{1}(\lambda)}{ \bar \alpha (\bar \lambda)}\right]
= - \frac{8 \bar e_0\bar \zeta^3_0}{\bar a''(\zeta_0)}\ n^{(+)}_{2}(\bar \lambda_0) e^{-2i\theta(\bar \lambda_0)}
\end{align}
and
\begin{align}\label{Res-barlambda-0}
\begin{split}
&\mathrm{Res}_{\lambda=\bar\lambda_{0}}\left[\frac{ n^{(-)}_{1}(\lambda)}{ \bar\alpha (\bar\lambda)}\right] \\
&=
- \frac{8 \bar e_0 \bar \zeta_0^3}{\bar a''(\zeta_0)} e^{-2i\theta(\bar \lambda_0)}
\left[ (n^{(+)}_{2})'(\bar \lambda_0) +  n^{(+)}_{2}(\lambda_0)
\left(-2 i \theta'(\bar \lambda_0) + \frac{\bar h_0}{2 \bar e_0 \bar{\zeta}_0} + \frac{1}{\bar \zeta_0^2} - \frac{\bar a'''(\zeta_0)}{6 \bar{\zeta}_0 \bar a''(\zeta_0)}\right)\right].
\end{split}
\end{align}
Using the same notations \eqref{A-0-B-0-definition} for $A_0$ and $B_0$, we transform \eqref{P-2-barlambda-0} and \eqref{Res-barlambda-0} into \eqref{P-2-bar-lambda-0-1} and \eqref{Res-bar-lambda-0-1}.
\end{proof}

\subsection{Computation of solutions of the linear algebraic system}

Using the first column of $M(\lambda)$ in (\ref{RH-solution-double}) for $\lambda \in \mathbb{C}_+$, we obtain from \eqref{P-2-bar-lambda-0-1}
and \eqref{Res-bar-lambda-0-1} that
\begin{align}
\label{double-M-1-lambda0}
\begin{split}
n^{(+)}_{1}(\lambda) &= n_1^{+\infty} e_1
- \frac{\bar{A}_0 \bar{\lambda}_0}{(\lambda-\bar{\lambda}_0)^2} \left[ 1 + (\lambda - \bar{\lambda}_0) (-2i \theta'(\bar \lambda_0) + \bar{B}_0 + \bar{\lambda}_0^{-1}) \right] n^{(+)}_{2}(\bar{\lambda}_0) e^{-2i \theta(\bar \lambda_0)} \\
& \qquad \qquad - \frac{\bar{A}_0 \bar{\lambda}_0}{\lambda - \bar{\lambda}_0} (n^{(+)}_2)'(\bar \lambda_0) e^{-2 i \theta(\bar \lambda_0)}.
\end{split}
\end{align}
Using the second column of $M(\lambda)$ in (\ref{RH-solution-double}) for $\lambda \in \mathbb{C}_-$, we obtain from   \eqref{P-2-lambda-0-1} and \eqref{Res-lambda-0-1} that
\begin{align}\label{double-M-2-barlambda0}
\begin{split}
n^{(+)}_{2}(\lambda) &= n_2^{+\infty} e_2
  + \frac{A_0}{(\lambda-\lambda_0)^2} \left[ 1 + (\lambda - \lambda_0) (2i \theta'(\lambda_0) + B_0) \right] n^{(+)}_{1}(\lambda_0) e^{2i \theta(\lambda_0)} \\
  & \qquad \qquad  + \frac{A_0}{\lambda - \lambda_0} (n^{(+)}_1)'(\lambda_0) e^{2 i \theta(\lambda_0)}.
  \end{split}
\end{align}
We can close the algebraic system by evaluating \eqref{double-M-1-lambda0} at $\lambda = \lambda_0$ and \eqref{double-M-2-barlambda0} at $\lambda = \bar{\lambda}_0$:
\begin{align}
\label{double-M-closed-1}
\begin{split}
n^{(+)}_{1}(\lambda_0) &= n_1^{+\infty} e_1
- \bar{C}_0 \bar{\lambda}_0 \left[ 1 + (\lambda_0 - \bar{\lambda}_0) (-2i \theta'(\bar \lambda_0) + \bar{B}_0 + \bar{\lambda}_0^{-1}) \right] n^{(+)}_{2}(\bar{\lambda}_0) e^{-2i \theta(\bar \lambda_0)} \\
& \qquad \qquad - \bar{C}_0 \bar{\lambda}_0 (\lambda_0 - \bar{\lambda}_0) (n^{(+)}_2)'(\bar \lambda_0) e^{-2 i \theta(\bar \lambda_0)}, \\
n^{(+)}_{2}(\bar\lambda_0) &= n_2^{+\infty} e_2
+ C_0 \left[ 1 - (\lambda_0 - \bar{\lambda}_0) (2i \theta'(\lambda_0) + B_0) \right] n^{(+)}_{1}(\lambda_0) e^{2i \theta(\lambda_0)} \\
& \qquad \qquad  - C_0 (\lambda_0 - \bar{\lambda}_0) (n^{(+)}_1)'(\lambda_0) e^{2 i \theta(\lambda_0)},
  \end{split}
\end{align}
as well as their derivatives at $\lambda = \lambda_0$ and $\lambda = \bar{\lambda}_0$ respectively:
\begin{align}\label{double-M-closed-2}
\begin{split}
(n^{(+)}_{1})'(\lambda_0) &=  \bar{C}_0 \bar{\lambda}_0 \left[ 2(\lambda_0 - \bar{\lambda}_0)^{-1} -2i \theta'(\bar \lambda_0) + \bar{B}_0 + \bar{\lambda}_0^{-1} \right] n^{(+)}_{2}(\bar{\lambda}_0) e^{-2i \theta(\bar \lambda_0)} \\
& \qquad \qquad + \bar{C}_0 \bar{\lambda}_0  (n^{(+)}_2)'(\bar \lambda_0) e^{-2 i \theta(\bar \lambda_0)}, \\
(n^{(+)}_{2})'(\bar{\lambda}_0) &= C_0 \left[ 2(\lambda_0 - \bar{\lambda}_0)^{-1} - 2i \theta'(\lambda_0) - B_0 \right] n^{(+)}_{1}(\lambda_0) e^{2i \theta(\lambda_0)} \\
& \qquad \qquad  - C_0 (n^{(+)}_1)'(\lambda_0) e^{2 i \theta(\lambda_0)},
  \end{split}
\end{align}
where
\begin{align*}
C_0 := \frac{A_0}{(\lambda_0-\bar{\lambda}_0)^2}.
\end{align*}

The following proposition solves the linear system (\ref{double-M-closed-1}) and (\ref{double-M-closed-2}) and derive the explicit representation for
the exponential double-soliton solution of the MTM system (\ref{MTM}) by using the recovery formulas (\ref{recover-potential}).

\begin{prop}\label{u-v-parameter-expression}
The potentials $u(x,t)$ and $v(x,t)$ in (\ref{recover-potential}) are expressed from solutions of the RH problem with a double pole by
	\begin{align}
	\label{cramer-u-v}
	u = \frac{\overline{N}_u}{\overline{D(M)}},\quad
	v = \frac{N_v}{D(M)},
	\end{align}
	where
	\begin{align}
	\label{tech-N-u}
	\begin{split}
	N_u &=- \lambda_0^{-1} A_0 e^{2i\theta(\lambda_0)} (2i \theta'(\lambda_0) + B_0- \lambda_0^{-1} ) - \lambda_0^{-1} A_0 |C_0|^2 \bar{\lambda}_0e^{4i \theta(\lambda_0)-2i \theta(\bar \lambda_0)} \\
	& \times [ 4 (\lambda_0 - \bar{\lambda}_0)^{-1} +(-2i \theta'(\bar{\lambda}_0) + \bar{B}_0 + \bar{\lambda}_0^{-1}) \bar{\lambda}_0 \lambda_0^{-1} -  3\lambda_0^{-1} ],
	\end{split}
	\end{align}
\begin{align}
	\begin{split}	
	N_v &= \bar A_0  e^{-2i \theta(\bar \lambda_0)} \left(-2i \theta'(\bar{\lambda}_0) + \bar{B}_0\right) +\bar A_0|C_0|^2 \bar{\lambda}_0e^{2i \theta(\lambda_0)-4i \theta(\bar \lambda_0)} \\
	& \times [ -4 (\lambda_0 - \bar{\lambda}_0)^{-1}
	+ (2i \theta'(\lambda_0) + B_0) \lambda_0 \bar{\lambda}_0^{-1} -  3 \bar{\lambda}_0^{-1} ].
		\end{split}
		\label{tech-N-v}
	\end{align}
	and
	\begin{align}
	\label{tech-D-M}
	\begin{split}
D(M) &= 1+|C_0|^4 \bar{\lambda}^2_0 e^{4i \theta(\lambda_0)-4i \theta(\bar \lambda_0)}\\
& \quad +|C_0|^2 \bar{\lambda}_0 e^{2i \theta(\lambda_0)-2i \theta(\bar \lambda_0)}
\left[ 6+2(\lambda_0 - \bar{\lambda}_0)(-2i \theta'(\bar{\lambda}_0) + \bar{B}_0 + \bar{\lambda}_0^{-1}-2i \theta'(\lambda_0) - B_0)\right.\\
&\qquad\qquad\left.-(\lambda_0 - \bar{\lambda}_0)^2(-2i \theta'(\bar{\lambda}_0) + \bar{B}_0 + \bar{\lambda}_0^{-1})(2i \theta'(\lambda_0) + B_0)\right]
	\end{split}
	\end{align}
\end{prop}

\begin{proof}	
By substituting \eqref{RH-solution-double}, \eqref{P-2-lambda-0-1}, and \eqref{Res-lambda-0-1} into
\eqref{recover-potential}, we obtain
\begin{align}
\label{recover-potential-u-double}
\begin{split}
\overline{u } &= \lim_{\lambda \to 0} M_{12}(\lambda) \\
&=-\frac{1}{\lambda_0} n_2^{+\infty} \mathrm{Res}_{\lambda=\lambda_0} \left[ \frac{ n^{(-)}_{21}(\lambda)}{ \alpha(\lambda)} \right] + \frac{1}{\lambda_0^2} n_2^{+\infty} \mathrm{P}^{-2}_{\lambda=\lambda_0} \left[ \frac{ n^{(-)}_{21}(\lambda)}{  \alpha(\lambda)} \right],\\
&=-\frac{1}{\lambda_0} n_2^{+\infty} A_0 e^{2i\theta(\lambda_0)}
\left[ (n^{(+)}_{11})'(\lambda_0) + n^{(+)}_{11}(\lambda_0) (2 i \theta'(\lambda_0) + B_0 - \lambda_0^{-1} ) \right],
\end{split}
\end{align}
where the second index for vectors $n_1^{(+)}$ and $n_2^{(-)}$ denotes the corresponding components of $2$-vectors.
Similarly, by substituting \eqref{RH-solution-double}, \eqref{P-2-bar-lambda-0-1}, and \eqref{Res-bar-lambda-0-1} into  \eqref{recover-potential}, we obtain
\begin{align}
\label{recover-potential-v-double}
\begin{split}
v &= \lim_{\lambda \to 0} M_{21}(\lambda)\\
&=-\frac{1}{\bar \lambda_0} n_1^{+\infty} \mathrm{Res}_{\lambda= \bar \lambda_0} \left[ \frac{ n^{(-)}_{12}(\lambda)}{\bar \alpha(\lambda)} \right] + \frac{1}{\bar \lambda_0^2} n_1^{+\infty} \mathrm{P}^{-2}_{\lambda= \bar \lambda_0} \left[ \frac{ n^{(-)}_{12}(\lambda)}{ \bar \alpha(\lambda)} \right],\\
&= n_1^{+\infty} \bar{A}_0 e^{-2i\theta(\bar \lambda_0)}
\left[ (n^{(+)}_{22})'(\bar \lambda_0) + n^{(+)}_{22}(\bar \lambda_0) (-2 i \theta'(\bar \lambda_0) + \bar{B}_0  ) \right].
\end{split}
\end{align}
The linear system (\ref{double-M-closed-1}) and (\ref{double-M-closed-2}) can be rewritten for the vectors $n_1^{(+)}(\lambda_0)$ and $(n_1^{(+)})'(\lambda_0)$:
\begin{align*}
M \left( \begin{matrix} n_1^{(+)}(\lambda_0) \\(n_1^{(+)})'(\lambda_0) \end{matrix} \right) = \left( \begin{matrix} n_1^{+\infty} e_1
- \bar{C}_0 \bar{\lambda}_0 \left[ 1 + (\lambda_0 - \bar{\lambda}_0) (-2i \theta'(\bar \lambda_0) + \bar{B}_0 + \bar{\lambda}_0^{-1}) \right] n_{2}^{+\infty}  e^{-2i \theta(\bar \lambda_0)} e_2  \\
\bar{C}_0 \bar{\lambda}_0 \left[ 2(\lambda_0 - \bar{\lambda}_0)^{-1} -2i \theta'(\bar \lambda_0) + \bar{B}_0 + \bar{\lambda}_0^{-1} \right] n_2^{+\infty} e^{-2i \theta(\bar \lambda_0)} e_2 \end{matrix} \right),
\end{align*}
where
\begin{align*}
M = \left( \begin{matrix} M_{11} I & M_{12} I \\ M_{21} I & M_{22} I \end{matrix} \right)
\end{align*}
with $I$ being a $2$-by-$2$ identity matrix and $M_{ij}$ being scalar entries given by
\begin{align*}
  M_{11} &=
1 + |C_0|^2 \bar{\lambda}_0 e^{2i \theta(\lambda_0)-2i \theta(\bar \lambda_0)}
[ 3 + (\lambda_0 - \bar{\lambda}_0) (-2i \theta'(\bar{\lambda}_0) + \bar{B}_0 + \bar{\lambda}_0^{-1}) - 2 (\lambda_0 - \bar{\lambda}_0) (2i \theta'(\lambda_0) + B_0) \\
& \qquad - (\lambda_0 - \bar{\lambda}_0)^2 (-2i \theta'(\bar{\lambda}_0) + \bar{B}_0 + \bar{\lambda}_0^{-1}) (2i \theta'(\lambda_0) + B_0) ], \\
M_{12}  & = -|C_0|^2 \bar{\lambda}_0e^{2i \theta(\lambda_0)-2i \theta(\bar \lambda_0)}(\lambda_0 - \bar{\lambda}_0)\left[ 2+(\lambda_0 - \bar{\lambda}_0)(-2i \theta'(\bar \lambda_0) + \bar{B}_0 + \bar{\lambda}_0^{-1})\right], \\
M_{21} &= -|C_0|^2 \bar{\lambda}_0 (\lambda_0 - \bar{\lambda}_0)^{-1} e^{2i \theta(\lambda_0)-2i \theta(\bar \lambda_0)}
[ 4 + (\lambda_0 - \bar{\lambda}_0) (-2i \theta'(\bar{\lambda}_0) + \bar{B}_0 + \bar{\lambda}_0^{-1}) \\
& \qquad - 3 (\lambda_0 - \bar{\lambda}_0)  (2i \theta'(\lambda_0) + B_0) - (\lambda_0 - \bar{\lambda}_0)^2 (2i \theta'(\lambda_0) + B_0) (-2i \theta'(\bar{\lambda}_0) + \bar{B}_0 + \bar{\lambda}_0^{-1}) ], \\
M_{22} &= 1 + |C_0|^2 \bar{\lambda}_0 e^{2i \theta(\lambda_0)-2i \theta(\bar \lambda_0)}
[ 3 + (\lambda_0 - \bar{\lambda}_0) (-2i \theta'(\bar{\lambda}_0) + \bar{B}_0 + \bar{\lambda}_0^{-1}) ].
\end{align*}
By Cramer's rule, we obtain the first components of vectors
$n^{(+)}_{1}(\lambda_0)$ and $(n^{(+)}_{1})'(\lambda_0)$:
\begin{align}
\label{Cramer-rule}
n^{(+)}_{11}(\lambda_0) = \frac{n_1^{\infty} M_{22} }{D(M)}, \qquad
(n^{(+)}_{11})'(\lambda_0) = \frac{-n_1^{\infty} M_{21}}{D(M)},
\end{align}
where $D(M) = M_{11}M_{22}-M_{12}M_{21}$ recovers (\ref{tech-D-M}) after cancelation of several terms at $|C_0|^4$.
Substituting (\ref{Cramer-rule}) into (\ref{recover-potential-u-double}), we get
$u$ in the form (\ref{cramer-u-v}) with
\begin{align*}
N_u &= -\lambda_0^{-1} A_0 e^{2i\theta(\lambda_0)}
	\left[ -M_{21} +  M_{22}  (2i \theta'(\lambda_0) + B_0 - \lambda_0^{-1} ) \right]
	\end{align*}
which yields (\ref{tech-N-u}) after cancelation of several terms at $|C_0|^2$.

For the vectors $n_2^{(+)}(\lambda_0)$ and $(n_2^{(+)})'(\lambda_0)$, the linear system (\ref{double-M-closed-1}) and (\ref{double-M-closed-2}) can be rewritten  as
\begin{align*}
\tilde M \left( \begin{matrix} n_2^{(+)}(\lambda_0) \\(n_2^{(+)})'(\lambda_0) \end{matrix} \right)
= \left( \begin{matrix} n_2^{+\infty} e_2
+ C_0 \left[ 1 - (\lambda_0 - \bar{\lambda}_0) (2i \theta'(\lambda_0) + B_0) \right] n_{1}^{+\infty}  e^{2i \theta(\lambda_0)} e_1  \\
C_0 \left[ 2(\lambda_0 - \bar{\lambda}_0)^{-1} - 2i \theta'(\lambda_0) - B_0 \right] n_1^{+\infty} e^{2i \theta(\lambda_0)} e_1 \end{matrix} \right),
\end{align*}
where
\begin{align*}
\tilde M = \left( \begin{matrix} \tilde M_{11} I & \tilde M_{12} I \\ \tilde M_{21} I & \tilde M_{22} I \end{matrix} \right)
\end{align*}
with
\begin{align*}
  \tilde M_{11} &=
1 + |C_0|^2 \bar{\lambda}_0 e^{2i \theta(\lambda_0)-2i \theta(\bar \lambda_0)}
[ 3 +  2(\lambda_0 - \bar{\lambda}_0) (-2i \theta'(\bar{\lambda}_0) + \bar{B}_0 + \bar{\lambda}_0^{-1}) -  (\lambda_0 - \bar{\lambda}_0) (2i \theta'(\lambda_0) + B_0) \\
& \qquad - (\lambda_0 - \bar{\lambda}_0)^2 (-2i \theta'(\bar{\lambda}_0) + \bar{B}_0 + \bar{\lambda}_0^{-1}) (2i \theta'(\lambda_0) + B_0) ], \\
\tilde M_{12} & =|C_0|^2 \bar{\lambda}_0e^{2i \theta(\lambda_0)-2i \theta(\bar \lambda_0)}(\lambda_0 - \bar{\lambda}_0)\left[ 2-(\lambda_0 - \bar{\lambda}_0)(2i \theta'(\lambda_0) + B_0)\right],\\
\tilde M_{21} &= |C_0|^2 \bar{\lambda}_0 (\lambda_0 - \bar{\lambda}_0)^{-1} e^{2i \theta(\lambda_0)-2i \theta(\bar \lambda_0)}
[ 4 + 3(\lambda_0 - \bar{\lambda}_0) (-2i \theta'(\bar{\lambda}_0) + \bar{B}_0 + \bar{\lambda}_0^{-1}) \\
& \qquad -  (\lambda_0 - \bar{\lambda}_0)  (2i \theta'(\lambda_0) + B_0) - (\lambda_0 - \bar{\lambda}_0)^2 (2i \theta'(\lambda_0) + B_0) (-2i \theta'(\bar{\lambda}_0) + \bar{B}_0 + \bar{\lambda}_0^{-1}) ], \\
\tilde M_{22} &= 1 + |C_0|^2 \bar{\lambda}_0 e^{2i \theta(\lambda_0)-2i \theta(\bar \lambda_0)}
[ 3 - (\lambda_0 - \bar{\lambda}_0) (2i \theta'(\lambda_0) + B_0) ].
\end{align*}
By Cramer's rule, we obtain the second components of vectors
$n_2^{(+)}(\lambda_0)$ and $(n_2^{(+)})'(\lambda_0)$:
\begin{align}
\label{Cramer-rule-v}
n^{(+)}_{22}(\lambda_0) =  \frac{n_2^{\infty} \tilde M_{22} }{D(\tilde M)}, \quad
(n^{(+)}_{22})'(\lambda_0) = \frac{-n_2^{\infty}\tilde  M_{21}}{D(\tilde M)},
\end{align}
where $D(\tilde M) = \tilde M_{11}\tilde M_{22}-\tilde M_{12}\tilde M_{21} = D(M)$ is given by (\ref{tech-D-M}). Substituting (\ref{Cramer-rule-v}) into (\ref{recover-potential-v-double}), we get $v$ in the form (\ref{cramer-u-v}) with
\begin{align*}
  N_v = \bar A_0  e^{-2i \theta(\bar \lambda_0)} \left[-\tilde{M}_{21} + \tilde{M}_{22} (-2 i \theta'(\bar{\lambda}_0) + \bar{B}_0) \right],
\end{align*}
which yields (\ref{tech-N-v}) after cancelation of several terms at $|C_0|^2$.
\end{proof}

\subsection{Proof of Theorem \ref{explicit-expression-NU-Muv}}

In order to rewrite the recovered potentials of Proposition \ref{u-v-parameter-expression} in the simplified form of Theorem \ref{explicit-expression-NU-Muv}, we set $\lambda_0 = e^{i \gamma} \in \mathbb{S}^1 \cap \mathbb{C}^+$ with $\gamma \in (0,\pi)$.
A more general solution is obtained with the Lorentz symmetry (\ref{MTM-Lorentz}).

Since $\lambda_0 = e^{i \gamma}$, we obtain
\begin{align*}
  2 i \theta(\lambda_0) &= -x \sin\gamma + i t\cos\gamma,\\
  4 i \theta'(\lambda_0) &= i(x+t) + i (x-t) e^{-2i\gamma},
\end{align*}
and complex conjugate for $-2i\theta(\bar{\lambda}_0)$ and $-4 i \theta'(\bar{\lambda}_0)$.
	
Let us define $A_0 = (\lambda_0 - \bar{\lambda}_0)^2 \lambda_0^{3/2} = -4 (\sin \gamma)^2 e^{\frac{3 i \gamma}{2}}$, which yields $C_0 = e^{\frac{3i\gamma}{2}}$. A more general solution with two translational parameters $x_0,t_0 \in \mathbb{R}$ can be obtained by the translational symmetry (\ref{MTM-symm}) or by including two additional parameters in $A_0 \in \mathbb{C}$. Then it follows from (\ref{tech-N-u}), (\ref{tech-N-v}), and (\ref{tech-D-M}) that
\begin{align*}
N_u &= 2i (\sin \gamma)^2 e^{-x \sin \gamma + i t  \cos \gamma + \frac{i \gamma}{2}} \left( x + t + (x-t) e^{-2i \gamma} - 2 i B_0 + 2 i e^{-i \gamma}  \right. \\
& \left. \quad
- e^{-2 x \sin \gamma - i \gamma} [ 4 (\sin \gamma)^{-1} + (x + t + (x-t) e^{2i \gamma} + 2 i \bar{B}_0 + 2 i e^{i \gamma}) e^{-2i \gamma} - 6 i e^{-i \gamma}] \right),\\
N_v &= 2i (\sin \gamma)^2 e^{-x \sin \gamma - i t  \cos \gamma - \frac{3 i \gamma}{2}} \left( x + t + (x-t) e^{2i \gamma} + 2 i \bar{B}_0  \right. \\
& \left. \quad -
e^{-2 x \sin \gamma - i \gamma} [ 4 (\sin \gamma)^{-1} + (x + t + (x-t)e^{-2i \gamma} - 2 i B_0) e^{2i \gamma} + 6 i e^{i \gamma}] \right),
\end{align*}
and
\begin{align*}
\begin{split}
D(M) &= 1 + e^{-4 x \sin \gamma - 2 i \gamma} + e^{-2 x \sin \gamma - i \gamma} \\
& \quad \times
\left[ 6 + 4 (\sin \gamma) (x + t + (x-t) \cos 2\gamma  + i(\bar{B}_0 - B_0 + e^{i \gamma}) ) \right.\\
&\qquad\qquad\left. + (\sin \gamma)^2 (x + t + (x-t) e^{2i \gamma} + 2i \bar{B}_0 + 2i e^{i \gamma})(x + t + (x-t) e^{-2i \gamma} - 2i B_0)\right].
\end{split}
\end{align*}

By using trigonometric identities, we reduce expressions for $N_u$, $N_v$ and $D(M)$ to the form:
\begin{align*}
N_u &= 4i (\sin \gamma)^2 e^{-x \sin \gamma + i t  \cos \gamma - \frac{i \gamma}{2}} \left( x \cos \gamma + i t \sin \gamma - i B_0 e^{i \gamma} + i \right. \\
& \left. \quad
- e^{-2 x \sin \gamma - i \gamma} [ 2 \cot \gamma + x \cos \gamma - i t \sin \gamma  + i \bar{B}_0 e^{-i \gamma}] \right),\\
N_v &= 4i (\sin \gamma)^2 e^{-x \sin \gamma - i t  \cos \gamma - \frac{i \gamma}{2}} \left( x \cos \gamma - i t \sin \gamma + i \bar{B}_0 e^{-i \gamma} \right. \\
& \left. \quad -
e^{-2 x \sin \gamma - i \gamma} [ 2 \cot \gamma + x \cos \gamma + i t \sin \gamma - i B_0 e^{i \gamma} + i] \right),\\
D(M) &= 1 + e^{-4 x \sin \gamma - 2 i \gamma} + 2 e^{-2 x \sin \gamma - i \gamma} \\
& \quad \times
\left[ 3 + 2 (\sin \gamma) (2 x (\cos \gamma)^2 + 2t (\sin \gamma)^2  + i(\bar{B}_0 - B_0 + e^{i \gamma}) ) \right.\\
&\qquad\qquad\left. + 2 (\sin \gamma)^2 (x \cos \gamma -i t \sin \gamma + i \bar{B}_0 e^{-i\gamma} + i)(x \cos \gamma + i t \sin \gamma - i B_0 e^{i\gamma}) \right].
\end{align*}
By selecting $B_0 = -i e^{-i \gamma} [ \tilde{x}_0 \cos \gamma + i \tilde{t}_0 \sin \gamma]$ with arbitrary parameters $\tilde{x}_0$ and $\tilde{t}_0$, we
obtain the same expressions for $N_u$ and $N_v$ as in Theorem \ref{explicit-expression-NU-Muv}. Regarding the expression for $D(M)$, we obtain
\begin{align*}
D(M) &= 1 + e^{-4 x \sin \gamma - 2 i \gamma} + 2 e^{-2 x \sin \gamma - i \gamma} \\
& \quad \times
\left[ 3 + 2 (\sin \gamma) [ 2 (x-\tilde{x}_0) (\cos \gamma)^2 + 2 (t-\tilde{t}_0) (\sin \gamma)^2  + i \cos \gamma - \sin \gamma ] \right.\\
&\qquad\qquad\left. + 2 (\sin \gamma)^2 [ (x-\tilde{x}_0) \cos \gamma -i (t-\tilde{t}_0) \sin \gamma  + i ] [(x-\tilde{x}_0) \cos \gamma + i (t-\tilde{t}_0) \sin \gamma)] \right].
\end{align*}
Expanding the bracket being $2 e^{-2x \sin \gamma - i \gamma}$ yields
\begin{align*}
& 3 - 2 (\sin \gamma)^2 + 2 i (\sin \gamma) (\cos \gamma) + 4 (\sin \gamma) (\cos \gamma)^2 (x - \tilde{x}_0) + 2 (\sin \gamma)^3 (t - \tilde{t}_0) \\
& \quad + 2 (\sin \gamma)^2 (\cos \gamma)^2 (x - \tilde{x}_0)^2 + 2i (\sin \gamma)^2 (\cos \gamma) (x - \tilde{x}_0) + 2 (\sin \gamma)^4 (t - \tilde{t}_0)^2 \\
=&  1 + 2 (\sin \gamma)^2 \left[  \cot \gamma + (x-\tilde{x}_0) \cos \gamma + \frac{i}{2} \right]^2 + 2 (\sin \gamma)^4 \left[t-\tilde{t}_0 + \frac{1}{2 \sin \gamma} \right]^2,
\end{align*}
which coincides with the expression for $D(M)$ in Theorem \ref{explicit-expression-NU-Muv}.

\subsection{Computations of eigenvectors and generalized eigenvectors}

An eigenvector of the linear system (\ref{eigenvector}) is given by
$\psi^{(+)}_{1}(\zeta_0)$ for $\zeta_0 = e^{\frac{i\gamma}{2}}$ with $\gamma \in (0,\pi)$, which decays exponentially as $|x| \to \infty$
due to \eqref{asymptotic-Lax-pair} and \eqref{double-psi-2-psi-1}.
By using the transformation \eqref{Transformation-1}, we obtain
\begin{align}\label{Eigen-explicite-expression}
\psi^{(+)}_{1}(\zeta_0)
= e^{i \theta(\lambda_0)} [T(v,\zeta_0)]^{-1} n_1^{(+)}(\lambda_0),
\end{align}
where
\begin{align*}
[T(v,\zeta_0)]^{-1} = \left(
\begin{array}{cc}
1 & 0 \\
-e^{-\frac{i\gamma}{2}} v & e^{-\frac{i\gamma}{2}}
\end{array}
\right)
\end{align*}
and $n_1^{(+)}(\lambda_0)$ with $\lambda_0 = e^{i \gamma}$ is obtained from the linear system in the proof of Proposition \ref{u-v-parameter-expression}.
By Cramer's rule, as in (\ref{Cramer-rule}), we obtain
\begin{align}\label{eigen-expre}
n^{(+)}_{1}(\lambda_0) = \frac{1}{D(M)} P^{\infty} \left( \begin{array}{c}
M_{22} \\ -\bar{C}_0 \bar{\lambda}_0
e^{-2i \theta(\bar{\lambda}_0)} (M_{22} b_1 + M_{12} b_2)
\end{array} \right),
\end{align}
where
\begin{align*}
b_1 &=  1 + (\lambda_0 - \bar{\lambda}_0) (-2i \theta'(\bar \lambda_0) + \bar{B}_0 + \bar{\lambda}_0^{-1}), \\
b_2 &=  (\lambda_0 - \bar{\lambda}_0)^{-1} (2 + (\lambda_0 - \bar{\lambda}_0) (-2i \theta'(\bar \lambda_0) + \bar{B}_0 + \bar{\lambda}_0^{-1})),
\end{align*}
and we recall that $P^{\infty} = {\rm diag}(n^{+\infty}_1,n^{+\infty}_2)$ with
\begin{align*}
n^{+\infty}_{1} =e^{\frac{i}{4}\int_{+\infty}^{x}(|u|^2+|v|^2) dy} =  \bar{n}^{+\infty}_{2}.
\end{align*}
By using the same definitions of $A_0$ and $B_0$ as in the proof of Theorem \ref{explicit-expression-NU-Muv}, we obtain
\begin{align*}
M_{12} &= -4i (\sin \gamma)^2 e^{-2x\sin\gamma} \left( \cot\gamma + (x-\tilde x_0) \cos\gamma - i (t-\tilde t_0) \sin\gamma \right), \\
M_{22} &= 1+ (\sin \gamma) e^{-2x\sin\gamma} \left( 3 \cot\gamma + 2 (x-\tilde x_0) \cos\gamma - 2 i (t-\tilde t_0) \sin\gamma - i\right)
\end{align*}
which yields the second component of the vector $n_1^{(+)}(\lambda_0)$ in the explicit form:
\begin{align*}
-\bar{C}_0 \bar{\lambda}_0
e^{-2i \theta(\bar{\lambda}_0)} (M_{22} b_1 + M_{12} b_2) &= -e^{-x\sin\gamma-it\cos\gamma -\frac{3i \gamma}{2}} \\
& \quad \times \left(  e^{i \gamma} + 2 \sin\gamma [ (x-\tilde x_0 ) \cos\gamma - i (t-\tilde t_0) \sin\gamma ] - e^{-2x \sin \gamma - 2 i \gamma}   \right).
\end{align*}

A generalized eigenvector of the linear system (\ref{eigenvector-generalized}) is given by
$(\psi^{(+)}_{1})'(\zeta_0)$ for $\zeta_0 = e^{\frac{i\gamma}{2}}$ with $\gamma \in (0,\pi)$, which decays exponentially as $|x| \to \infty$
due to \eqref{asymptotic-Lax-pair}, \eqref{double-psi-2-psi-1}, and \eqref{psi-2-psi+1+psi+1}. By differentiating the transformation \eqref{Transformation-1} in $\lambda = \zeta^2$ at $\lambda_0 = \zeta_0^2 = e^{i \gamma}$, we obtain
\begin{align}
\label{Gener-Eigen-explicite-expression}
\begin{split}
(\psi^{(+)}_{1})'(\zeta_0) &= 2 \zeta_0 e^{i \theta(\lambda_0)} [T(v,\zeta_0)]^{-1} (n_1^{(+)})'(\lambda_0) + 2i \zeta_0 \theta'(\lambda_0)
\psi^{(+)}_1(\zeta_0) \\
& \qquad - \zeta_0^{-1} \partial_{\zeta} T(v,\zeta) \psi^{(+)}_1(\zeta_0),
\end{split}
 \end{align}
where $\psi^{(+)}_1(\zeta_0)$ is given by the exponentially decaying eigenfunction (\ref{Eigen-explicite-expression}) and
\begin{equation*}
\partial_{\zeta} T(v,\zeta) =
\left( \begin{matrix} 0 & 0 \\ 0 & 1 \end{matrix} \right).
\end{equation*}
Hence, in order to obtain $(\psi^{(+)}_{1})'(\zeta_0)$, we only need to compute $(n_1^{(+)})'(\lambda_0)$ and use the transformation (\ref{Gener-Eigen-explicite-expression}).
By Cramer's rule, as in (\ref{Cramer-rule}), we obtain
\begin{align}\label{generalized-eigen-expre}
(n^{(+)}_{1})'(\lambda_0) = \frac{1}{D(M)} P^{\infty} \left( \begin{array}{c}
-M_{21} \\ \bar{C}_0 \bar{\lambda}_0
e^{-2i \theta(\bar{\lambda}_0)} (M_{11} b_2 + M_{21} b_1)
\end{array} \right).
\end{align}
Proceeding similarly, we obtain
\begin{align*}
M_{11} &= 1 + e^{-2x\sin\gamma - i\gamma} \left( 3 +
2 (\sin \gamma) e^{i \gamma} [(x-\tilde x_0) \cos\gamma - i (t-\tilde t_0) \sin\gamma + i ] \right. \\
& \quad + 4 (\sin \gamma) e^{-i \gamma} [(x-\tilde x_0) \cos\gamma + i (t-\tilde t_0) \sin\gamma] \\
& \quad \left. + 4 (\sin \gamma)^2 [(x-\tilde x_0) \cos\gamma - i (t-\tilde t_0) \sin\gamma + i ] [(x-\tilde x_0) \cos\gamma + i (t-\tilde t_0) \sin\gamma] \right), \\
M_{21} &= \frac{i}{\sin \gamma} e^{-2x\sin\gamma-i\gamma} \left( 2 + (\sin \gamma) e^{i \gamma} [(x-\tilde x_0) \cos\gamma - i (t-\tilde t_0) \sin\gamma + i ] \right. \\
& \quad + 3 (\sin \gamma) e^{-i \gamma} [(x-\tilde x_0) \cos\gamma + i (t-\tilde t_0) \sin\gamma] \\
& \quad \left. + 2 (\sin \gamma)^2 [(x-\tilde x_0) \cos\gamma - i (t-\tilde t_0) \sin\gamma + i ] [(x-\tilde x_0) \cos\gamma + i (t-\tilde t_0) \sin\gamma] \right), \\
\end{align*}
which yields the second component of the vector $(n_1^{(+)})'(\lambda_0)$ in the explicit form:
\begin{align*}
& \bar{C}_0 \bar{\lambda}_0 e^{-2i \theta(\bar{\lambda}_0)} (M_{11} b_2 + M_{21} b_1) =  -i e^{-x\sin\gamma-it\cos\gamma-\frac{3i}{2}\gamma}  \left( \cot \gamma  + (x-\tilde x_0 ) \cos\gamma - i (t-\tilde t_0) \sin\gamma   \right. \\
& \qquad \left. +  e^{-2x\sin\gamma -3i\gamma} \left[ \cot \gamma  +  (x-\tilde x_0 ) \cos\gamma + i (t-\tilde t_0) \sin\gamma + i   \right] \right).
\end{align*}

\begin{rem}
	The explicit expressions (\ref{eigen-expre}) and (\ref{generalized-eigen-expre}) confirm that both the eigenvector (\ref{Eigen-explicite-expression}) and the generalized eigenvector (\ref{Gener-Eigen-explicite-expression}) decays exponentially as $x \to \pm \infty$ since $D(M) \to 1$ as $x \to +\infty$ and $D \sim e^{-4 x \sin \gamma -2 i \gamma}$ as $x \to -\infty$.
\end{rem}

\subsection{Numerical illustration of the exponential double-solitons}

We plot the exponential double-soliton solutions of Theorem \ref{explicit-expression-NU-Muv} in Figure \ref{fig-1} for three different values of $\gamma$. The translational parameters in (\ref{double-soliton}) are set to $\tilde{x}_0 = \frac{1}{\sin\gamma}$ and $ \tilde{t}_0 = \frac{1}{2\sin\gamma}$. The solutions describe scattering of two identical solitons which slowly approach to each other, overlap, and then slowly diverge from each other.

\begin{figure}[htb!]
	\centering
	\includegraphics[width=5.5cm,height=5.3cm,angle=0]{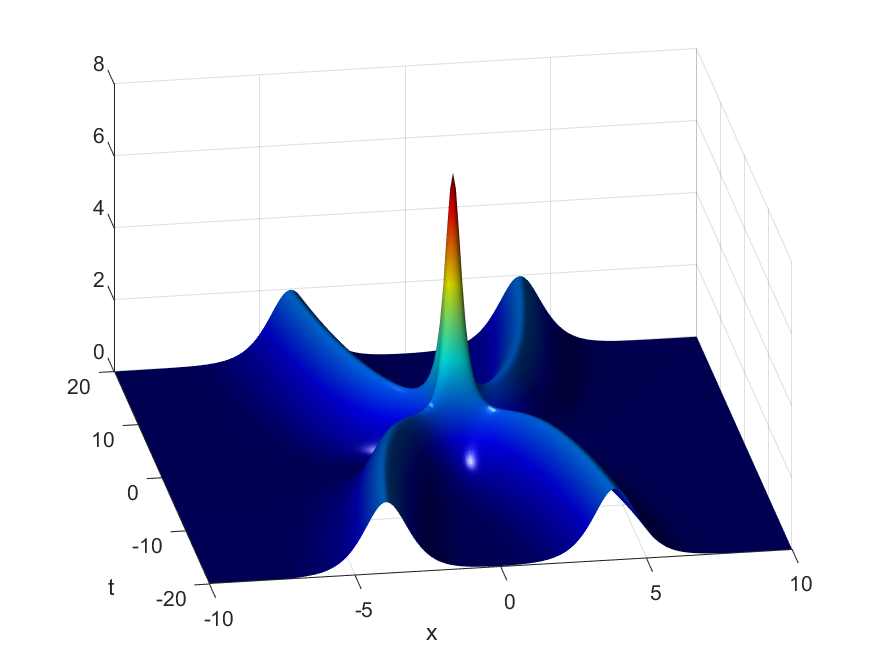}
	\includegraphics[width=5.5cm,height=5.3cm,angle=0]{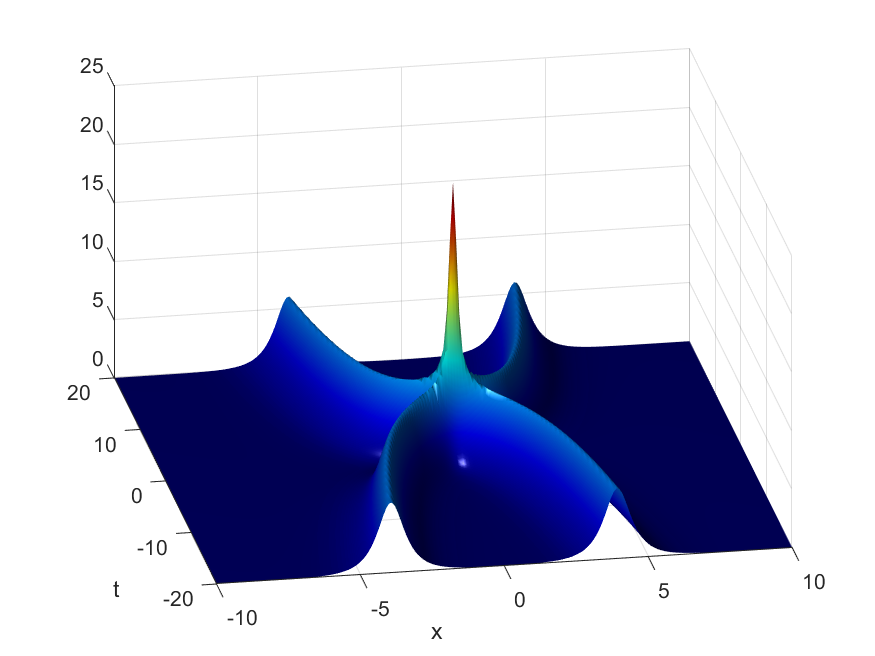}
	\includegraphics[width=5.5cm,height=5.3cm,angle=0]{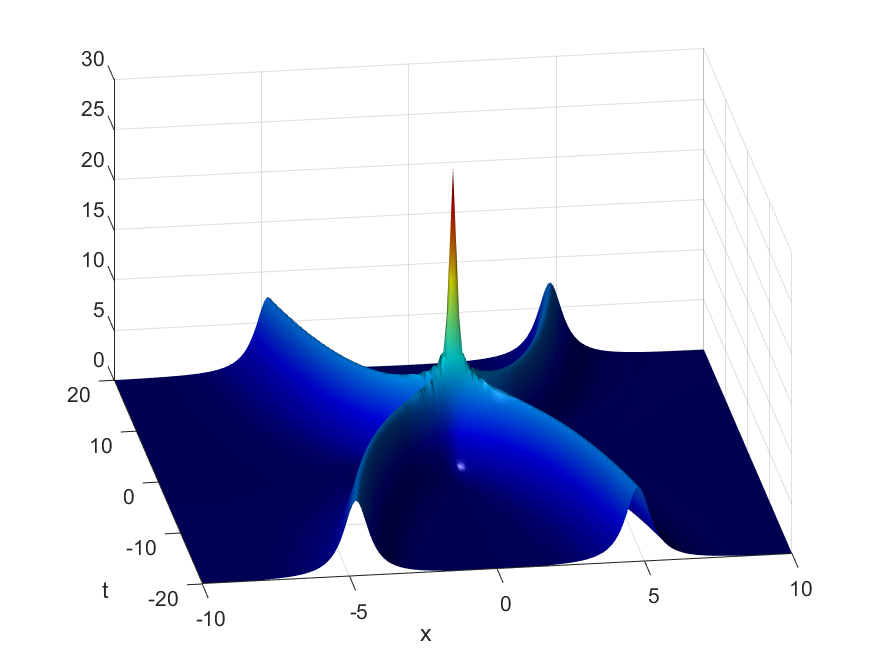}\\
	(\textbf{a})\qquad  \qquad\qquad\qquad\qquad\qquad~~~(\textbf{b})
	\qquad\qquad\qquad\qquad\qquad(\textbf{c})\\
	
	\caption{(Color online) The surface plots of $|u(x,t)|^2 + |v(x,t)|^2$ for the exponential double-soliton solutions with (a) $\gamma=\frac{\pi}{3}$, (b) $\gamma=\frac{2\pi}{3}$, and (c) $\gamma=\frac{5\pi}{6}$.}\label{fig-1}
\end{figure}

We shall find the approximate distance between the two identical solitons
for large $|x| + |t|$. It follows from the bilinear equations,
see \cite{Chen-SAPM-2023,Han-2024} and Appendix B, that
\begin{align}
  |u|^2+|v|^2=\frac{ |N_u|^2+|N_v|^2}{|D(M)|^2} = 2i \frac{\partial}{\partial x} \log \frac{\overline{D(M)}}{D(M)}.
\end{align}
Therefore, we just need to investigate the behavior of $D(M)$  for large $|x| + |t|$. The dominant terms of $D(M)$ as $|x| + |t| \to \infty$ are given by
\begin{align*}
D(M) \sim e^{-2x \sin \gamma - i \gamma} \left( e^{-2x \sin \gamma - i \gamma}
+ 4 (\sin \gamma)^4 t^2 \right),
\end{align*}
from which we obtain that
\begin{align}
\label{law}
 |x | \sim \frac{ \ln |t|}{ \sin \gamma}, \quad \mbox{\rm as} \;\; |x| + |t| \to  \infty.
\end{align}
The dependence (\ref{law}) is shown in Figure \ref{fig-2} by red line
together with the contour plots from Figure \ref{fig-1}.

\begin{figure}[H]
	\centering
	\includegraphics[width=5.3cm,height=4.8cm,angle=0]{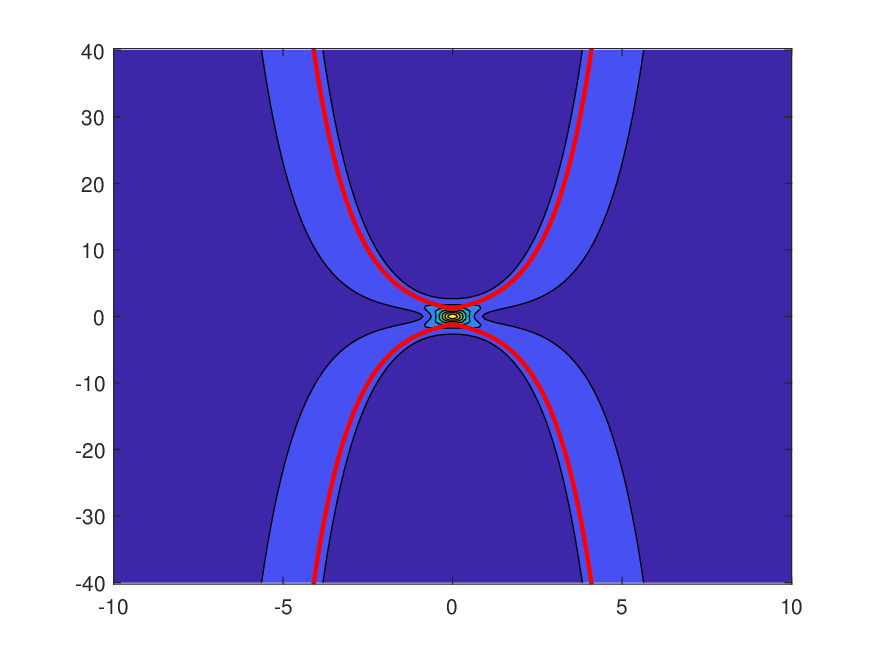}
	\includegraphics[width=5.3cm,height=4.8cm,angle=0]{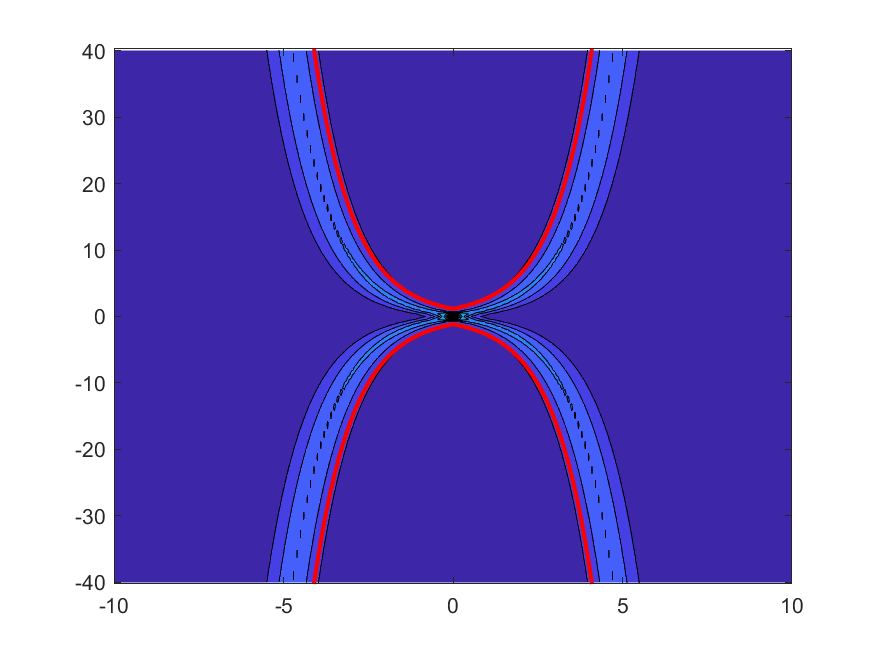}
	 \includegraphics[width=5.3cm,height=4.8cm,angle=0]{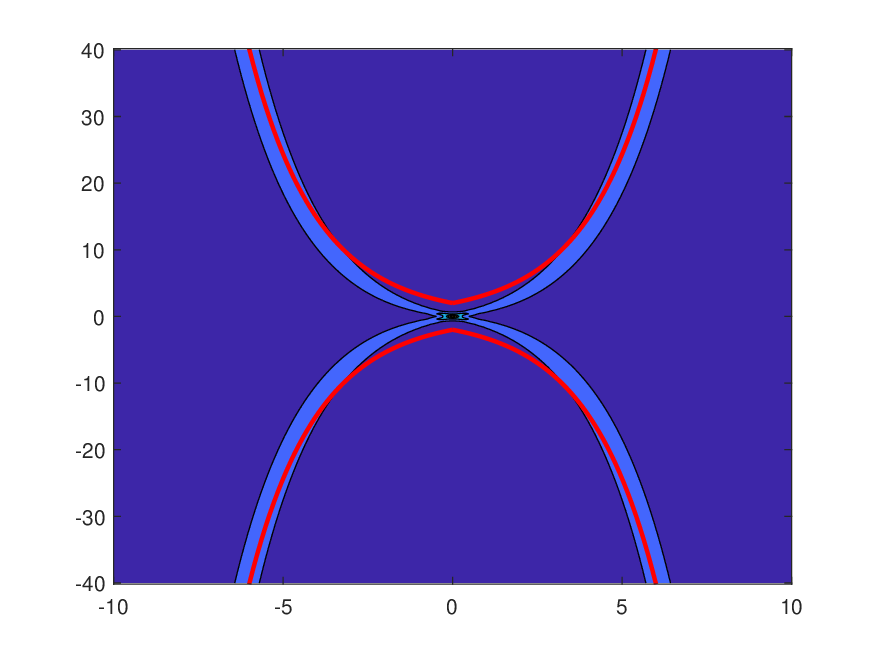}
	
	(\textbf{a})\qquad  \qquad\qquad\qquad\qquad\qquad~~~(\textbf{b})
	 \qquad\qquad\qquad\qquad\qquad(\textbf{c})\\

	\caption{\small (Color online) The contour plots of $|u(x,t)|^2+|v(x,t)|^2$ for the solutions of Figure \ref{fig-1} with (a) $\gamma=\frac{\pi}{3}$, (b) $\gamma=\frac{2\pi}{3}$, (c) $\gamma=\frac{5\pi}{6}$.}
	\label{fig-2}
\end{figure}

\section{Limit to the algebraic double-solitons}
\label{sec-4}

Here we take the limit $\gamma \to \pi$ of the exponential double-solitons in Theorem \ref{explicit-expression-NU-Muv} to derive the algebraic double-solitons. We show that the algebraic double-solitons correspond
to the double embedded eigenvalue $\zeta_0 = i$ in the linear systems (\ref{eigenvector}) and (\ref{eigenvector-generalized}). In order to obtain nontrivial limits, we change the arbitrary parameters $\tilde x_0$ and $\tilde t_0$ used in
Section \ref{sec-3} with the transformation
\begin{equation}
\label{shift-phase}
\tilde x_0 \to \tilde x_0+\frac{1}{\sin\gamma}, \qquad
\tilde t_0 \to \tilde t_0+\frac{1}{2\sin\gamma}.
\end{equation}
The two computations below give the proof of Theorem \ref{Theorem-algebra-eigen-general}.

\subsection{Computations of (\ref{u-algebric-soliton}) and (\ref{v-algebric-soliton})}

Let $\gamma:=\pi-\varepsilon$ and consider the limit $\varepsilon\to0^{+}$.  Taylor's expansions yield
\begin{align*}
\sin\gamma &= \varepsilon - \frac{1}{6} \varepsilon^3 + \mathcal{O}(\varepsilon^5),\\
\cos\gamma &= -1 + \frac{\varepsilon^2}{2} + \mathcal{O}(\varepsilon^4).
\end{align*}
To obtain (\ref{u-algebric-soliton}) and (\ref{v-algebric-soliton}), we only need to substitute (\ref{shift-phase}) into $D(M)$, $N_u$ and $N_v$ given below (\ref{double-soliton}) and collect together the coefficients of Taylor expanion at powers $\varepsilon$, $\varepsilon^2$, $\varepsilon^3$, and $\varepsilon^4$.
With the transformation (\ref{shift-phase}), we rewrite $D(M)$ as
\begin{align*}
D(M) &= 1 + e^{-4 x \sin \gamma - 2 i \gamma} + 2 e^{-2 x \sin \gamma - i \gamma} \\
& \quad \times
\left( 1 + 2 (\sin \gamma)^2 \left[  (x-\tilde{x}_0) \cos \gamma + \frac{i}{2} \right]^2 + 2 (\sin \gamma)^4 (t-\tilde{t}_0)^2 \right).
\end{align*}
Expansion as $\varepsilon \to 0$ gives nonzero terms at powers $\varepsilon^2$, $\varepsilon^3$, and $\varepsilon^4$:
\begin{align*}
\varepsilon^2:\qquad & -4 \tilde x_0 (i - 2 x + \tilde{x}_0), \\
\varepsilon^3:\qquad & -4 \tilde x_0 (i-2x) (i - 2 x + \tilde{x}_0), \\
\varepsilon^4:\qquad & \frac{1}{12} (i-2x)^4 + \frac{1}{3} (i-2x)(i-6x) - 4(t-\tilde{t}_0)^2 - 2 \tilde x_0 (i-2x)^2 (i - 2 x + \tilde{x}_0) \\
& \quad
+ \frac{2}{3} \tilde{x}_0 (8\tilde{x}_0-16x+5i).
\end{align*}

With the transformation (\ref{shift-phase}), we rewrite $N_u$ and $N_v$ as
\begin{align*}
N_u &= 4i (\sin \gamma)^2 e^{-x \sin \gamma + i t  \cos \gamma - \frac{i \gamma}{2}} \left( -\cot \gamma + (x-\tilde{x}_0) \cos \gamma + i  (t-\tilde{t}_0) \sin \gamma + \frac{i}{2} \right. \\
& \left. \quad
- e^{-2 x \sin \gamma - i \gamma} \left[ \cot \gamma + (x-\tilde{x}_0) \cos \gamma - i (t-\tilde{t}_0) \sin \gamma + \frac{i}{2} \right] \right)
\end{align*}
and
\begin{align*}
N_v &= 4i (\sin \gamma)^2 e^{-x \sin \gamma - i t  \cos \gamma - \frac{i \gamma}{2}} \left( - \cot \gamma + (x-\tilde{x}_0)  \cos \gamma - i (t-\tilde{t}_0) \sin \gamma + \frac{i}{2} \right. \\
& \left. \quad
- e^{-2 x \sin \gamma - i \gamma} \left[ \cot \gamma + (x-\tilde{x}_0)  \cos \gamma + i (t-\tilde{t}_0) \sin \gamma + \frac{i}{2} \right] \right).
\end{align*}
The expression in the circular brackets for $N_u$ and $N_v$ are multiplied by
$$
4i (\sin \gamma)^2 e^{-x \sin \gamma \pm i t  \cos \gamma - \frac{i \gamma}{2}} \sim 4 \varepsilon^2 e^{\mp it} \quad \mbox{\rm as} \quad \varepsilon \to 0.
$$
Since the expansions of the exponential factors in powers of $\varepsilon$ do not modify the limit of $\varepsilon \to 0$, we collect nonzero terms in the expansions of $N_u e^{-it\cos \gamma}$ and $N_v e^{it \cos \gamma}$ at powers $\varepsilon^2$, $\varepsilon^3$, and $\varepsilon^4$:
\begin{align*}
\varepsilon^2:\qquad & 8  \tilde x_0, \\
\varepsilon^3:\qquad &   4 \tilde x_0 (i-2x), \\
\varepsilon^4:\qquad &  \frac{1}{3} (i-2x)^3 + 5 \tilde{x}_0 (i-2x)^2 - 4 i (t - \tilde{t}_0) (i - 2 x) + \frac{4 i}{3} - \frac{20}{3}  \tilde{x}_0,
\end{align*}
and
 \begin{align*}
\varepsilon^2:\quad &  8  \tilde x_0, \\
\varepsilon^3:\quad & 4 \tilde x_0 (i-2x), \\
\varepsilon^4:\quad & \frac{1}{3} (i-2x)^3 + 5 \tilde{x}_0 (i-2x)^2 + 4 i (t - \tilde{t}_0) (i - 2 x) + \frac{4 i}{3} - \frac{20}{3} \tilde{x}_0.
\end{align*}
By rescaling $\tilde{x}_0 \to \tilde{x}_0 \varepsilon^2$, we now obtain the nontrivial limit at the power $\varepsilon^4$:
\begin{align*}
\lim_{\varepsilon \to 0} D(M) \varepsilon^{-4} &= \frac{4}{3} x^4 -\frac{8}{3}  i x^3  + 2 x^2   - 2 i x -\frac{1}{4} - 4(t - \tilde{t}_0)^2 - 4 \tilde{x}_0 (i - 2 x),\\
\lim_{\varepsilon \to 0} N_u e^{it} \varepsilon^{-4}  &= -\frac{8}{3} x^3
+ 4i x^2 + 2x + i - 4i (t - \tilde{t}_0) (i - 2 x) + 8 \tilde{x}_0, \\
\lim_{\varepsilon \to 0} N_v e^{-it} \varepsilon^{-4}  &=  -\frac{8}{3} x^3
+ 4i x^2 + 2x + i + 4i (t - \tilde{t}_0) (i - 2 x) + 8 \tilde{x}_0,
\end{align*}
which yields the explicit expressions (\ref{u-algebric-soliton}) and (\ref{v-algebric-soliton}) from the quotients given by (\ref{double-soliton}).

\subsection{Computations of (\ref{calim-eigenvector-expression}) and (\ref{calim-generalized-eigenvector-expression})}

We substitute the phase shift (\ref{shift-phase}) into $ n_{1}^{(+)}(\lambda_0)$ and $ (n_{1}^{(+)})'(\lambda_0)$ given in (\ref{eigen-expre}) and (\ref{generalized-eigen-expre}). These expressions
define the eigenvector and the generalized eigenvector of the linear systems (\ref{eigenvector}) and (\ref{eigenvector-generalized}) for $\zeta_0 = e^{\frac{i \gamma}{2}}$ by
(\ref{Eigen-explicite-expression}) and (\ref{Gener-Eigen-explicite-expression}) respectively. By using $\gamma := \pi-\varepsilon$ and expanding in powers of $\varepsilon$, we derive (\ref{calim-eigenvector-expression}) and (\ref{calim-generalized-eigenvector-expression}).

After the transformation (\ref{shift-phase}), we rewrite  $ M_{22}$ and  $-\bar{C}_0 \bar{\lambda}_0
e^{-2i \theta(\bar{\lambda}_0)} (M_{22} b_1 + M_{12} b_2)$
in (\ref{eigen-expre}) as follows:
\begin{align*}
M_{22}= 1+ (\sin \gamma) e^{-2x\sin\gamma} \left(  \cot\gamma + 2 (x-\tilde x_0) \cos\gamma - 2 i (t-\tilde t_0) \sin\gamma \right)
\end{align*}
and
\begin{align*}
&-\bar{C}_0 \bar{\lambda}_0
e^{-2i \theta(\bar{\lambda}_0)} (M_{22} b_1 + M_{12} b_2) = -e^{-x\sin\gamma-it\cos\gamma -\frac{3i \gamma}{2}}\\
 &\qquad \qquad \times \left(  e^{i \gamma} + 2 \sin\gamma \left[ -\cot\gamma+(x-\tilde x_0 ) \cos\gamma - i (t-\tilde t_0) \sin\gamma +\frac{i}{2} \right] - e^{-2x \sin \gamma - 2 i \gamma}   \right).
\end{align*}
Expansion as $\varepsilon \to 0$ gives nonzero terms at powers $\varepsilon^1$, and $\varepsilon^2$.
 \begin{itemize}
   \item The coefficients of $ M_{22}$:
   \begin{align*}
   \varepsilon^1: \quad &  2 \tilde x_0,\\
     \varepsilon^2:\quad  & 2x^2-2i (t - \tilde{t}_0) +\frac{1}{2} -4x\tilde x_0.
   \end{align*}
   \item The coefficients of $-\bar{C}_0 \bar{\lambda}_0
e^{-2i \theta(\bar{\lambda}_0)} (M_{22} b_1 + M_{12} b_2) e^{it \cos \gamma}$:
   \begin{align*}
  \varepsilon^1: \quad & -2i \tilde x_0,\\
  \varepsilon^2: \quad & 2ix^2+ 4 x-2 (t - \tilde{t}_0) -\frac{3}{2} i + 3\tilde x_0+2ix\tilde x_0.
   \end{align*}
 \end{itemize}
Rescaling $\tilde{x}_0 \to \tilde{x}_0 \varepsilon^2$ yields the nontrivial limit at the power $\varepsilon^2$:
\begin{align*}
\lim_{\varepsilon \to 0}  & M_{22} \varepsilon^{-2}= 2x^2-2i (t - \tilde{t}_0) +\frac{1}{2},\\
\lim_{\varepsilon \to 0} & \left( -\bar{C}_0 \bar{\lambda}_0
e^{-2i \theta(\bar{\lambda}_0)} (M_{22} b_1 + M_{12} b_2) e^{-it} \right) \varepsilon^{-2}  = -2ix^2- 4 x+2 (t - \tilde{t}_0) +\frac{3}{2} i.
\end{align*}
By using (\ref{Eigen-explicite-expression}) and (\ref{eigen-expre}), taking the limit
\begin{align*}
\psi_0 := \lim_{\varepsilon\to 0}\varepsilon^2 \psi^{(+)}_{1}(\zeta_0),
\end{align*}
we obtain (\ref{calim-eigenvector-expression}) with
$$
\mathfrak{n}_0 := \lim_{\varepsilon\to 0}\varepsilon^2 n^{(+)}_{1}(\lambda_0)
$$
given below (\ref{calim-eigenvector-expression}) in the explicit form.

To obtain a nontrivial limit for the generalized eigenvector,
we rewrite the expression \eqref{Gener-Eigen-explicite-expression} in the equivalent form:
\begin{align*}
&(\psi^{(+)}_{1})'(\zeta_0)-2(\lambda_0-\bar\lambda_0)^{-1}\psi^{(+)}_{1}(\zeta_0) \\
&= 2 \zeta_0 e^{i \theta(\lambda_0)} [T(v,\zeta_0)]^{-1}\left[  (n_1^{(+)})'(\lambda_0) +  \left(i  \theta'(\lambda_0)-2(\lambda_0-\bar\lambda_0)^{-1}\right)
n_1^{(+)}(\lambda_0) \right] \\
& \qquad - \zeta_0^{-1} \partial_{\zeta} T(v,\zeta) \psi^{(+)}_1(\zeta_0),
 \end{align*}
After the transformation (\ref{shift-phase}), we rewrite  $ -M_{21}$ and  $\bar{C}_0 \bar{\lambda}_0 e^{-2i \theta(\bar{\lambda}_0)} (M_{11} b_2 + M_{21} b_1)$ in \eqref{generalized-eigen-expre} as follows:
\begin{align*}
-M_{21} = -2ie^{-2x\sin\gamma- i\gamma}\left[ \sin\gamma \left[(x-\tilde x_0)^2(\cos\gamma)^2+ (t-\tilde t_0)^2(\sin\gamma)^2  \right] + i(t-\tilde t_0)\sin\gamma\cos\gamma+ \frac{  \sin\gamma}{4} \right]
\end{align*}
and
\begin{align*}
& \bar{C}_0 \bar{\lambda}_0 e^{-2i \theta(\bar{\lambda}_0)} (M_{11} b_2 + M_{21} b_1)=  -i e^{-x\sin\gamma-it\cos\gamma-\frac{3i}{2}\gamma}  \\
& \qquad \times \left(    (x-\tilde x_0 ) \cos\gamma - i (t-\tilde t_0) \sin\gamma +\frac{i}{2}  +  e^{-2x\sin\gamma -3i\gamma} \left[    (x-\tilde x_0 ) \cos\gamma + i (t-\tilde t_0) \sin\gamma +  \frac{i}{2}   \right] \right).
\end{align*}
Expansion as $\varepsilon \to 0$ gives nonzero terms at powers $\varepsilon^0$, $\varepsilon^1$, and $\varepsilon^2$ for both components of the numerator of  $ (n_1^{(+)})'(\lambda_0) +  \left(i  \theta'(\lambda_0)- 2(\lambda_0-\bar\lambda_0)^{-1}\right)
n_1^{(+)}(\lambda_0)$.
\begin{itemize}
   \item The coefficients of $-M_{21} +  \left(i  \theta'(\lambda_0)-2(\lambda_0-\bar\lambda_0)^{-1}\right)
M_{22}$:
   \begin{align*}
   \varepsilon^0: \quad & - 2i \tilde x_0,\\
   \varepsilon^1: \quad & ix\tilde x_0+2i \tilde x_0^2,\\
   \varepsilon^2: \quad & \frac{4}{3}i \tilde x_0+x (t - \tilde t_0) -\frac{1}{3}i x^3+ 2i x^2 \tilde x_0- \frac{4}{3}i x +\tilde x_0t\\
   & +3x \tilde  x_0-4i x \tilde  x_0^2-2x^2 -2\tilde x_0^2-\frac{1}{2}+2i(t - \tilde t_0).
   \end{align*}
   \item The coefficients of
   $$
   \bar{C}_0 \bar{\lambda}_0 e^{-2i \theta(\bar{\lambda}_0)} \left[ M_{11} b_2 + M_{21} b_1 -
   \left(i  \theta'(\lambda_0)-2(\lambda_0-\bar\lambda_0)^{-1}\right)
(M_{22} b_1 + M_{12} b_2) \right] e^{it \cos \gamma} :
$$
   \begin{align*}
   e^{it}\varepsilon^2: \quad & - 2 \tilde x_0,\\
   e^{it}\varepsilon^3: \quad  & 6i \tilde x_0-5 x \tilde x_0,\\
   e^{it} \varepsilon^4: \quad &   ix (t - \tilde t_0) - it \tilde x_0 -ix^2+ \frac{ 5}{4}i- \frac{1}{3}x^3- \frac{15}{4}x\\
    &+3(t - \tilde t_0)  +\frac{155}{12}\tilde x_0 -6x^2 \tilde x_0-it \tilde x_0 +\frac{35}{2}i x \tilde x_0.
   \end{align*}
 \end{itemize}
Rescaling $\tilde{x}_0 \to \tilde{x}_0 \varepsilon^2$ yields the nontrivial limit at the power $\varepsilon^2$:
\begin{align*}
&\lim_{\varepsilon \to 0}   \left[ -M_{21} +  \left(i  \theta'(\lambda_0)-2(\lambda_0-\bar\lambda_0)^{-1}\right)
M_{22} \right]\varepsilon^{-2}\\
&= -\frac{1}{3}i x^3-2x^2 - \frac{4}{3}i x+ x (t - \tilde t_0) +2i (t - \tilde t_0)-\frac{1}{2}- 2i \tilde x_0,\\
&\lim_{\varepsilon \to 0} \bar{C}_0 \bar{\lambda}_0 e^{-2i \theta(\bar{\lambda}_0)} \left[ M_{11} b_2 + M_{21} b_1 -
\left(i  \theta'(\lambda_0)-2(\lambda_0-\bar\lambda_0)^{-1}\right)
(M_{22} b_1 + M_{12} b_2) \right] e^{-it}  \varepsilon^{-2} \\
 &= - \frac{1}{3}x^3-ix^2- \frac{15}{4}x +ix (t - \tilde t_0) + \frac{ 5}{4}i + 3 (t - \tilde t_0)  - 2 \tilde x_0.
\end{align*}
By using (\ref{Gener-Eigen-explicite-expression}) and (\ref{generalized-eigen-expre}), taking the limit
\begin{align*}
\psi_1 := \lim_{\varepsilon\to 0}\varepsilon^2 \left[ (\psi^{(+)}_{1})'(\zeta_0)-2(\lambda_0-\bar\lambda_0)^{-1} \psi^{(+)}_{1}(\zeta_0)\right],
\end{align*}
we obtain (\ref{calim-generalized-eigenvector-expression}) with
$$
\mathfrak{n}_1  := \lim_{\varepsilon\to 0}\varepsilon^2 \left[ (n_1^{(+)})'(\lambda_0) +  \left(i  \theta'(\lambda_0)- 2(\lambda_0-\bar\lambda_0)^{-1}\right)
n_1^{(+)}(\lambda_0) \right]
$$
given below (\ref{calim-generalized-eigenvector-expression}) in the explicit form.

\appendix

\section{Single-soliton from a simple pole}
\label{app-A}

Here we consider solutions of the normalized RH problem for the refelectionless potential $r_{\pm}(\lambda) \equiv 0$ for $\lambda \in \R$ with a simple pole at $\lambda_0 \in \mathbb{C}^+$. By symmetry (\ref{symmetry-psi}),
both $\lambda_0 \in \mathbb{C}^+$ and $\bar{\lambda}_0 \in \mathbb{C}^-$ are
poles of $M(\lambda)$ in $\mathbb{C}$. The normalized RH problem
can be rewritten in the form:

\vspace{0.2cm}

\centerline{\fbox{\parbox[cs]{1.0\textwidth}{
			{\bf RH problem.} Find a complex-valued function $M(\lambda)$ with the following properties:
			\begin{itemize}
			\item $ M(\lambda)$ has simple poles at $\lambda_0 \in \mathbb{C}^+$ and $\bar{\lambda}_0 \in \mathbb{C}^-$.
			\item $ M(\lambda) \to \mathbb I$ as $|\lambda| \to \infty$,
			where $\mathbb I$ is the $2$-by-$2$ identity matrix.
			\item $ M_{+}(\lambda)= M_{-}(\lambda) $ for every $\lambda\in\mathbb{R}$, where $M_{\pm}(\lambda) := \lim\limits_{{\rm Im}(\lambda) \to \pm 0} M(\lambda)$.
			\end{itemize}
}}}

\vspace{0.2cm}

The solution of the RH problem is immediately given by
	\begin{align}
	\label{solution-RH}
	M(x,t,\lambda)=\mathbb I+\frac{{\rm Res}_{\lambda=\lambda_0} M_{+}(\lambda)}{\lambda-\lambda_0}+\frac{{\rm Res}_{\lambda=\bar\lambda_0} M_{-}(\lambda)}{\lambda-\bar\lambda_0}.
	\end{align}
In order to compute the residue terms, we note from (\ref{P-matrix}) that
$\lambda_0$ is a simple zero of $\alpha(\lambda)$ extended to $\mathbb{C}^+$ by Lemma \ref{lemma-2}. Since  it follows from \eqref{SMmatrix-1} that
\begin{align*}
\alpha(\lambda) = a(\zeta) = \det\left(\psi^{(+)}_{1}(\zeta),\psi^{(-)}_{2}(\zeta)\right),
\end{align*}
we define $\zeta_0 := \sqrt{\lambda_0}$ and a constant $b_0 \in \mathbb{C}$ such that the columns of $\psi^{(\pm)}(\zeta)$ satisfying (\ref{asymptotic-Lax-pair})
are related at $\zeta = \zeta_0$ by
	\begin{align}\label{eigenvector-simple-psi}
	\psi^{(-)}_{2}(\zeta_0) = b_0\psi^{(+)}_{1}(\zeta_0).
	\end{align}
Since $\lambda_0 \in \mathbb{C}^+$, it follows from (\ref{asymptotic-Lax-pair}) and (\ref{eigenvector-simple-psi})  that $\psi^{(-)}_{2}(\zeta_0)$ decays to zero exponentially fast both as $x \to \pm \infty$. Hence it is the eigenvector of the linear system \eqref{Lax-pair} for $\zeta = \zeta_0$.
	
By using the transformation \eqref{Transformation-1}, we can rewrite  \eqref{eigenvector-simple-psi} in the form:
	\begin{align}\label{tilde-n-2-b-n-1}
	n^{(-)}_{2}(\lambda_0)=b_0\zeta_0^{-1} n^{(+)}_{1}(\lambda_0) e^{2i\theta(\lambda_0)},
	\end{align}
	where $\theta(\lambda)$ is given by (\ref{theta}). By using (\ref{P-matrix}) and (\ref{tilde-n-2-b-n-1}), we compute the residue term as follows
\begin{align}\label{single-residue-lambda0}
{\rm Res}_{\lambda=\lambda_0} P(\lambda) = \left( \vec{\bf 0} \quad \frac{ n^{(-)}_{2}(\lambda_0)}{  \alpha'(\lambda_0)}\right)
= \left(\vec{\textbf{0}}\quad  \frac{b_0}{\zeta_0 \alpha'(\lambda_0)} n^{(+)}_{1}(\lambda_0)e^{2i\theta(\lambda_0)}\right),
\end{align}
where $\vec{\textbf{0}}$ is the $2$-by-$1$ null vector.

By using the symmetry condition (\ref{symmetry-psi}), we have
\begin{align*}
  \psi_1^{(\pm)}(\zeta)= \left(
                           \begin{array}{cc}
                             0 & 1 \\
                             -1 & 0 \\
                           \end{array}
                         \right) \overline{\psi}_2^{(\pm)}(\zeta),\quad
     \psi_2^{(\pm)}(\zeta)= \left(
                           \begin{array}{cc}
                             0 & -1 \\
                             1 & 0 \\
                           \end{array}
                         \right) \overline{\psi}_1^{(\pm)}(\zeta),
\end{align*}
then (\ref{eigenvector-simple-psi}) can be transformed into
	\begin{align*}
	\psi^{(-)}_{1}(\bar\zeta_0) = -\bar b_0\psi^{(+)}_{2}(\bar\zeta_0).
	\end{align*}
Using the transformation  (\ref{Transformation-1}), we  obtain
\begin{align}
\label{sym-other}
n^{(-)}_{1}(\bar\lambda_0) = -\bar b_0 \bar\zeta_0 n^{(+)}_{2}(\bar\lambda_0) e^{-2i\theta(\bar\lambda_0)},
\end{align}
from which we compute the other residue term by using (\ref{P-matrix}) and (\ref{sym-other}):
\begin{align}
\label{single-residue-conjecture-lambda0}
{\rm Res}_{\lambda=\bar{\lambda}_0} P(\lambda) = \left(  \frac{ n^{(-)}_{1}(\bar{\lambda}_0)}{  \bar{\alpha}'(\bar{\lambda}_0)} \quad \vec{\bf 0} \right)
= \left(-\frac{\bar{b}_0 \bar{\zeta}_0}{\bar{\alpha}'(\bar{\lambda}_0)} n^{(+)}_{2}(\bar{\lambda}_0) e^{-2i\theta(\bar{\lambda}_0)} \quad
\vec{\textbf{0}} \right),
\end{align}

We recall that $M(\lambda) = [P^{\infty}]^{-1} P(\lambda)$
with $P^{\infty} = {\rm diag}(n_1^{+\infty},n_2^{+\infty})$ with $n_2^{+\infty} = \overline{n_1^{+\infty}}$, see (\ref{p-infty}). Using the first column of (\ref{solution-RH}) at $\lambda = \lambda_0$ due to (\ref{single-residue-lambda0}) and the second column of (\ref{solution-RH})
at $\lambda = \bar{\lambda}_0$ due to (\ref{single-residue-conjecture-lambda0}), we obtain a closed system of linear algebraic equations:
\begin{align}
\label{evaluate-conjecture-lambda-0}
n^{(+)}_{1}(\lambda_0) = n_1^{+\infty} e_1
- \frac{\bar\lambda_0 \bar{c}_0}{\lambda_0-\bar\lambda_0} n^{(+)}_{2}(\bar\lambda_0)e^{-2i\theta(\bar\lambda_0)}
\end{align}
and
\begin{align}
\label{evaluate-lambda-0}
n^{(+)}_{2}(\bar\lambda_0) = n_2^{+\infty} e_2
+ \frac{c_0 }{\bar\lambda_0-\lambda_0} n^{(+)}_{1}(\lambda_0)e^{2i\theta(\lambda_0)},
\end{align}
where $e_1 = (1,0)^T$, $e_2 = (0,1)^T$, and
$$
c_0 := \frac{b_0}{\zeta_0 \alpha'(\zeta_0)}.
$$
Then, from \eqref{evaluate-conjecture-lambda-0} and \eqref{evaluate-lambda-0}, we have
\begin{align*}
	n^{(+)}_{1}(\lambda_0) &=  n^{+\infty}_1 e_1
	- \frac{\bar{\lambda}_0 \bar{c}_0}{\lambda_0 - \bar{\lambda}_0} n_2^{+\infty} e^{-2i \theta(\bar{\lambda}_0)} e_2
	+\frac{\bar\lambda_0 |c_0|^2}{(\lambda_0-\bar\lambda_0)^2} e^{2i\theta(\lambda_0)-2i\theta(\bar\lambda_0)} n^{(+)}_{1}(\lambda_0),\\
	n^{(+)}_{2}(\bar\lambda_0)&=  n^{+\infty}_2 e_2 +
	\frac{c_0 }{\bar\lambda_0-\lambda_0} n^{+\infty}_{1} e^{2i\theta(\lambda_0)} e_1 + \frac{\bar\lambda_0 |c_0|^2}{(\lambda_0-\bar\lambda_0)^2}
	e^{2i\theta(\lambda_0)-2i\theta(\bar\lambda_0)} n^{(+)}_{2}(\bar\lambda_0).
\end{align*}
By using (\ref{recover-potential}), we obtain the explicit solutions
to the MTM system (\ref{MTM}) in the form
\begin{align*}
u = \lim_{\lambda \to 0} \overline{M}_{12}(\lambda) = - \frac{\bar{c}_0}{\bar{\lambda}_0}  n_1^{+\infty} \overline{n^{(+)}_{11}(\lambda_0)} e^{-2i\theta(\bar{\lambda}_0)} = - \frac{\bar{c}_0 e^{-2i \theta(\bar{\lambda}_0)}}{\bar{\lambda}_0 - \frac{|\lambda_0|^2 |c_0|^2}{(\lambda_0 - \bar{\lambda}_0)^2}  e^{2i \theta(\lambda_0)-2i\theta(\bar{\lambda}_0)}}
\end{align*}
and
\begin{align*}
v = \lim_{\lambda \to 0} M_{21}(\lambda) =\bar{c}_0 n_1^{+\infty}(x) n^{(+)}_{22}(\bar\lambda_0) e^{-2i\theta(\bar\lambda_0)} = \frac{\bar{c}_0 e^{-2i \theta(\bar{\lambda}_0)}}{1 - \frac{\bar{\lambda}_0 |c_0|^2}{(\lambda_0 - \bar{\lambda}_0)^2}  e^{2i \theta(\lambda_0)-2i\theta(\bar{\lambda}_0)}}.
\end{align*}

To simplify the expressions for the single-soliton solution $(u,v)$, we pick
$\lambda_0 = e^{i \gamma}$ with $\gamma \in (0,\pi)$ on
$\mathbb{S}^1 \cap \mathbb{C}^+$. A more general solution can be obtained with the Lorentz symmetry (\ref{MTM-Lorentz}). If $\lambda_0 = e^{i \gamma}$, we obtain from (\ref{theta}) that
	\begin{align*}
2 i \theta(\lambda_0) = -\alpha x + i \beta t,
	\end{align*}
	where $\alpha = \sin \gamma$ and $\beta = \cos \gamma$.
In addition, we choose
$$
c_0 = 2i \sin \gamma e^{\frac{i \gamma}{2}}
$$	
and obtain the single-soliton solution in the form
\begin{align}
\label{sol-u}
u(x,t) = 2i \alpha \frac{e^{-\alpha x - i \beta t + \frac{i \gamma}{2}}}{1 + e^{-2 \alpha x + i \gamma}} = i \alpha \; {\rm sech}\left(\alpha x - \frac{i \gamma}{2}\right) \; e^{-i \beta t}
\end{align}
and
\begin{align}
\label{sol-v}
v(x,t) = -2i  \alpha \; \frac{e^{-\alpha x - i \beta t - \frac{i \gamma}{2}}}{1 + e^{-2 \alpha x - i \gamma}} = -i \alpha \; {\rm sech}\left(\alpha x +  \frac{i\gamma}{2}\right) \; e^{-i \beta t},
\end{align}
which coincides with (\ref{one_soliton_1}) since $\alpha = \sin \gamma$ and $\beta = \cos \gamma$.
A more general solution with two translational parameters $x_0,t_0 \in \R$ can be obtained by using the symmetries (\ref{MTM-symm}) or by introducing two translational parameters in the expression for $c_0 \in \mathbb{C}$.
	
Finally, we write the explicit form of the eigenvector $\psi_1^{(+)}(\zeta_0)$, see (\ref{eigenvector-simple-psi}), which satisfies (\ref{eigenvector}) with $(u,v)$ given by (\ref{sol-u})--(\ref{sol-v}) and with $\zeta_0 = e^{\frac{i \gamma}{2}}$. By using the transformation \eqref{Transformation-1}, we write
	\begin{align*}
\psi^{(+)}_{1}(\zeta_0)	= e^{i\theta(\lambda_0)} [T(v,\zeta_0)]^{-1} n_1^{(+)}(\lambda_0),
\end{align*}
where
\begin{align*}
[T(v,\zeta_0)]^{-1} = \left( \begin{matrix} 1  & 0 \\
-e^{-\frac{i \gamma}{2}} v & e^{-\frac{i\gamma}{2}} \end{matrix} \right)
\end{align*}
and
\begin{align*}
n_1^{(+)}(\lambda_0) &= \frac{n^{+\infty}_1 e_1
- \frac{\bar{\lambda}_0 \bar{c}_0}{\lambda_0 - \bar{\lambda}_0} n_2^{+\infty} e^{-2i \theta(\bar{\lambda}_0)} e_2}{1 - \frac{\bar\lambda_0 |c_0|^2}{(\lambda_0-\bar\lambda_0)^2} e^{2i\theta(\lambda_0)-2i\theta(\bar\lambda_0)}} \\
&= \frac{1}{1 + e^{-2 \alpha x - i \gamma}}
\left( \begin{array}{c} n_1^{+\infty} \\
e^{-\alpha x - i \beta t - \frac{3 i \gamma}{2}} n_2^{+\infty} \end{array} \right) \\
&= \frac{1}{2} {\rm sech}\left(\alpha x + \frac{i \gamma}{2}\right)
\left( \begin{array}{c} e^{\alpha x + \frac{i \gamma}{2} + \frac{i}{4} \int_{+\infty}^x (|u|^2 + |v|^2) dx} \\
e^{- i \beta t - i \gamma - \frac{i}{4} \int_{+\infty}^x (|u|^2 + |v|^2) dx} \end{array} \right).
\end{align*}
We note that
$$
\frac{1}{4} \int_{\mathbb{R}} (|u|^2 + |v|^2) dx =
 \int_{\mathbb{R}} \frac{\sin^2 \gamma}{\cosh(2 \sin \gamma \; x) + \cos\gamma} dx = \gamma.
$$
Hence we can write
\begin{equation}
\label{eig-simple-isol}
e^{i\theta(\lambda_0)} n_1^{(+)}(\lambda_0) = \frac{1}{2} {\rm sech}\left(\alpha x + \frac{i \gamma}{2}\right) \left( \begin{array}{c} e^{\frac{1}{2} \alpha x + \frac{i}{2} \beta t + \frac{i \gamma}{2} + \frac{i}{4} \int_{+\infty}^x (|u|^2 + |v|^2) dx} \\
e^{-\frac{1}{2} \alpha x - \frac{i}{2} \beta t - \frac{i}{4} \int_{-\infty}^x (|u|^2 + |v|^2) dx}\end{array} \right),
\end{equation}
which decays exponentially to $0$ as $x \to \pm \infty$. Since $[T(v,\zeta_0)]^{-1}$ is bounded, then $\psi^{(+)}_{1}(\zeta_0) \in H^1(\mathbb{R},\mathbb{C}^2)$ is an exponentially decaying eigenvector of the linear system (\ref{eigenvector}).

The algebraic soliton appears in the singular limit $\gamma \to \pi$, where $\zeta_0 \to i$. The simple eigenvalue $\zeta_0 = i$ is embedded into the continuous spectrum of the Lax system (\ref{Lax-pair}), which is located on $\R \cup (i \R)$. By writing $\gamma := \pi - \epsilon$ and taking the limit $\epsilon \to 0$ in (\ref{sol-u}) and (\ref{sol-v}), we obtain
\begin{equation}
\label{alg-sol}
u(x,t) = \frac{2i}{1 - 2 ix} e^{it}, \quad
v(x,t) = -\frac{2i}{1+2ix} e^{it}.
\end{equation}
The eigenvector $\psi_0$ for the simple embedded eigenvalue $\zeta_0 = i$
of the linear system (\ref{eigenvector}) with $(u,v)$ given by (\ref{alg-sol}) is obtained from (\ref{eig-simple-isol}) in the limit $\epsilon \to 0$ in the explicit form:
\begin{align}
\label{eig-simple-emb}
\psi_0 := \lim_{\epsilon \to 0} \epsilon \psi^{(+)}_{1}(\zeta_0) = \frac{1}{1 + 2 ix}  \left( \begin{matrix} 1  & 0 \\
i v & -i \end{matrix} \right) \left( \begin{array}{c} e^{- \frac{i t}{2} + i \arctan(2x)} \\
e^{\frac{it}{2} - \frac{i \pi}{2} - i \arctan(2x)}\end{array} \right),
\end{align}
where we have used the elementary integral
$$
\frac{1}{4} \int_{-\infty}^x (|u|^2 + |v|^2) dx =
\int_{-\infty}^x \frac{2}{1+4x^2} dx = \frac{\pi}{2} + \arctan(2x).
$$
Based on the explicit expression (\ref{eig-simple-emb}), we confirm
that $\psi_0 \in H^1(\mathbb{R},\mathbb{C}^2)$ is an algebraically
decaying eigenvector of the linear system (\ref{eigenvector}) such that
$|\psi_0(x)| = \mathcal{O}(|x|^{-1})$ as $|x| \to \infty$.

\section{Exponential double-solitons in the bilinear Hirota method}
\label{app-B}

Here we obtain the exponential double-soliton solutions
by using the bilinear Hirota method developed in \cite{Chen-SAPM-2023}.
To proceed with computations, we use the parameterization from \cite{Han-2024} and write the general exponential two-soliton solutions in the form:
\begin{align}
	\label{exp-double-sol}
	u = \frac{g}{\bar{f}}, \quad v=\frac{h}{f},
	\end{align}
	where
	\begin{align*}
	f &= 1 + e^{-2 \xi_1 - \mathrm{i} \gamma_1} + e^{-2 \xi_2 - \mathrm{i} \gamma_2}
	+ A_{12} e^{-2 \xi_1 - 2 \xi_2 - \mathrm{i} \gamma_1 -\mathrm{i} \gamma_2} - 4 \sqrt{\delta_1 \delta_2} \sin \gamma_1 \sin \gamma_2
	e^{-\xi_1 - \xi_2 - \frac{\mathrm{i}}{2} \gamma_1 - \frac{\mathrm{i}}{2} \gamma_2} \notag \\
	& \quad
	\times
	\left[ \frac{\delta_1 e^{-\mathrm{i} (\eta_1 - \eta_2)}}{(\delta_1 e^{-\frac{\mathrm{i}}{2}(\gamma_1 + \gamma_2)}  - \delta_2  e^{\frac{\mathrm{i}}{2}(\gamma_1 + \gamma_2)})^2} +
	\frac{\delta_2 e^{\mathrm{i} (\eta_1 - \eta_2)}}{(\delta_1 e^{\frac{\mathrm{i}}{2}(\gamma_1 + \gamma_2)}  - \delta_2  e^{-\frac{\mathrm{i}}{2}(\gamma_1 + \gamma_2)})^2} \right],
	\end{align*}	
	\begin{align*}
	h &= -\bar{\alpha}_1 e^{-\xi_1 - \mathrm{i} \eta_1} \left[ 1 + \left( \frac{p_1 - p_2}{p_1 + \bar{p}_2} \right)^2 e^{-2 \xi_2 + \mathrm{i} \gamma_2} \right]
	- \bar{\alpha}_2 e^{-\xi_2 - \mathrm{i} \eta_2} \left[ 1 + \left( \frac{p_1 - p_2}{\bar{p}_1 + p_2} \right)^2 e^{-2 \xi_1 + \mathrm{i} \gamma_1} \right],
	\end{align*}
	and
	\begin{align*}
	g &= \frac{\mathrm{i} \bar{\alpha}_1}{p_1} e^{-\xi_1 - \mathrm{i} \eta_1}
	\left[ 1 + \left( \frac{p_1 - p_2}{p_1 + \bar{p}_2} \right)^2 e^{-2 \xi_2 + 3 \mathrm{i} \gamma_2} \right]
	+\frac{\mathrm{i} \bar{\alpha}_2}{p_2} e^{-\xi_2 - \mathrm{i} \eta_2} \left[ 1 + \left( \frac{p_1 - p_2}{\bar{p}_1 + p_2} \right)^2 e^{-2 \xi_1 + 3 \mathrm{i} \gamma_1} \right],
	\end{align*}
with arbitrary parameters $\gamma_j \in (0, \pi)$, $\delta_j > 0$,  $(x_j,t_j) \in \mathbb{R}^2$, and uniquely defined for $j = 1,2$ as
\begin{align*}
	p_j = \mathrm{i} \delta_j e^{-i \gamma_j}, \quad
\alpha_j = 2 \sqrt{\delta_j} \sin \gamma_j e^{\frac{\mathrm{i} \gamma_j}{2}},
\end{align*}
	\begin{align*}
	\xi_j &= \sin \gamma_j \left( \frac{1}{2} (\delta_j + \delta_j^{-1}) x
	+ \frac{1}{2} (\delta_j-\delta_j^{-1}) t + x_j\right), \\
	\eta_j &= \cos \gamma_j \left( \frac{1}{2} (\delta_j - \delta_j^{-1}) x
	+ \frac{1}{2} (\delta_j + \delta_j^{-1}) t  + t_j\right),
	\end{align*}
	and
	\begin{align*}
	A_{12} = \left( \frac{\delta_1^2 + \delta_2^2 - 2 \delta_1 \delta_2 \cos(\gamma_1 - \gamma_2)}{\delta_1^2 + \delta_2^2 - 2 \delta_1 \delta_2 \cos(\gamma_1 + \gamma_2)} \right)^2.
	\end{align*}
Due to Lorentz transformation (\ref{MTM-Lorentz}), we can consider
the exponential double-solitons with zero speed, for which we take
$\delta_1 = \delta_2 = 1$. In addition, we use translational symmetry and replace $e^{-\xi_{1,2}}$ by $e^{-\xi_{1,2}} \sin \left(\frac{\gamma_1 + \gamma_2}{2}\right)$ in all expressions. \\

\underline{Considering $f$,} we obtain
\begin{align*}
f &= 1 + \sin^2\left(\frac{\gamma_1+\gamma_2}{2}\right) \left[ e^{-2 \xi_1 - \mathrm{i} \gamma_1} + e^{-2 \xi_2 - \mathrm{i} \gamma_2} \right]
+ \sin^4\left(\frac{\gamma_1-\gamma_2}{2}\right)   e^{-2 \xi_1 - 2 \xi_2 - \mathrm{i} \gamma_1 -\mathrm{i} \gamma_2} \\
& \qquad + 2 \sin \gamma_1 \sin \gamma_2
e^{-\xi_1 - \xi_2 - \frac{\mathrm{i}}{2} \gamma_1 - \frac{\mathrm{i}}{2} \gamma_2} \cos(\eta_1 - \eta_2).
\end{align*}	
We now define the small parameter $\epsilon$ from
$\gamma_1 = \gamma + \epsilon$ and $\gamma_2 = \gamma - \epsilon$
and take the limit $\epsilon \to 0$ for a given $\gamma \in (0,\pi)$. In order to get a nontrivial limit, we also define the translational parameters
$x_{1,2}$ from the power series:
\begin{align*}
\xi_1 &= (\sin \gamma_1) (x + x_1) = \log(\epsilon) + \alpha (x - x_0) + \epsilon \beta (x - \tilde{x}_0)
-\frac{1}{2} \epsilon^2 \alpha (x - \tilde{\tilde{x}}_0) + \mathcal{O}(\epsilon^3), \\
\xi_2 &= (\sin \gamma_2) (x + x_2) = \log(\epsilon) + \alpha (x - x_0) - \epsilon \beta (x - \tilde{x}_0)
-\frac{1}{2} \epsilon^2 \alpha (x - \tilde{\tilde{x}}_0) + \mathcal{O}(\epsilon^3),
\end{align*}
with new translational parameters $x_0, \tilde{x}_0, \tilde{\tilde{x}}_0 \in \mathbb{R}$ and with $\alpha = \sin \gamma$, $\beta = \cos \gamma$. Similarly, we define the translational parameters $t_{1,2}$ from the power series:
\begin{align*}
\eta_1 &= (\cos \gamma_1) (t + t_1) = -\frac{\pi}{2} + \beta (t - t_0) - \epsilon \alpha (t - \tilde{t}_0)
-\frac{1}{2} \epsilon^2 \beta (t - \tilde{\tilde{t}}_0) + \mathcal{O}(\epsilon^3), \\
\eta_2 &= (\cos \gamma_2) (t + t_2) = \frac{\pi}{2} +  \beta (t - t_0) + \epsilon \alpha (t - \tilde{t}_0)
-\frac{1}{2} \epsilon^2 \beta (t - \tilde{\tilde{t}}_0) + \mathcal{O}(\epsilon^3),
\end{align*} 	
with new translational parameters $t_0, \tilde{t}_0, \tilde{\tilde{t}}_0 \in \mathbb{R}$. With these choices, we expand the expression for $f$ in powers of $\epsilon$ and obtain the following explicit expression
\begin{align}
	\lim_{\epsilon \to 0} f = 1 + e^{-2 \xi -i\gamma} \left[ 2 + \alpha^2
	(2 \beta (x - \tilde{x}_0) + i)^2 + 4 \alpha^4 (t - \tilde{t}_0)^2 \right] + e^{-4 \xi - 2i\gamma},
	  \label{f-3}
\end{align}
where $\xi := \alpha (x-x_0)$. \\

\underline{Considering $h$,} we obtain
	\begin{align*}
h &= -2 \sin \gamma_1 \sin\left(\frac{\gamma_1+\gamma_2}{2}\right) e^{-\xi_1 - \mathrm{i} \eta_1 - \frac{\mathrm{i} \gamma_1}{2} } \left[ 1 + \left( \frac{e^{-\mathrm{i} \gamma_1} - e^{-\mathrm{i} \gamma_2}}{e^{-\mathrm{i} \gamma_1} - e^{\mathrm{i} \gamma_2}} \right)^2 \sin^2\left(\frac{\gamma_1+\gamma_2}{2}\right) e^{-2 \xi_2 + \mathrm{i} \gamma_2} \right] \\
& \quad
- 2 \sin \gamma_2 \sin\left(\frac{\gamma_1+\gamma_2}{2}\right) e^{-\xi_2 - \mathrm{i} \eta_2 - \frac{\mathrm{i} \gamma_2}{2} }\left[ 1 + \left( \frac{e^{-\mathrm{i} \gamma_1} - e^{-\mathrm{i} \gamma_2}}{e^{-\mathrm{i} \gamma_2} - e^{\mathrm{i} \gamma_1}} \right)^2 \sin^2\left(\frac{\gamma_1+\gamma_2}{2}\right) e^{-2 \xi_1 + \mathrm{i} \gamma_1} \right],
\end{align*}
With the choice of the translational parameters above, we expand the expression for $h$ in the powers in $\epsilon$ and obtain the following explicit expression
\begin{align}
\lim_{\epsilon \to 0} h &=
4 \mathrm{i} \alpha^2 e^{-\xi - \mathrm{i} \eta - \frac{\mathrm{i} \gamma}{2}}
\left[ -\cot \gamma + \beta (x - \tilde{x}_0) - \mathrm{i} \alpha (t - \tilde{t}_0) + \frac{\mathrm{i}}{2} \right. \notag \\
	& \qquad \left. - e^{-2 \xi - \mathrm{i} \gamma}
	\left( \cot \gamma +
	\beta (x - \tilde{x}_0) + \mathrm{i} \alpha (t - \tilde{t}_0) + \frac{\mathrm{i}}{2} \right) \right],
\label{h-1}
\end{align}
where $\eta := \beta (t - t_0)$. \\

\underline{Considering $g$,} we obtain
\begin{align*}
g &= 2 \sin \gamma_1 \sin\left(\frac{\gamma_1+\gamma_2}{2}\right) e^{-\xi_1 - \mathrm{i} \eta_1 + \frac{\mathrm{i} \gamma_1}{2} } \left[ 1 + \left( \frac{e^{-\mathrm{i} \gamma_1} - e^{-\mathrm{i} \gamma_2}}{e^{-\mathrm{i} \gamma_1} - e^{\mathrm{i} \gamma_2}} \right)^2 \sin^2\left(\frac{\gamma_1+\gamma_2}{2}\right) e^{-2 \xi_2 + 3 \mathrm{i} \gamma_2} \right] \\
& \quad
+ 2 \sin \gamma_2 \sin\left(\frac{\gamma_1+\gamma_2}{2}\right) e^{-\xi_2 - \mathrm{i} \eta_2 + \frac{\mathrm{i} \gamma_2}{2} }\left[ 1 + \left( \frac{e^{-\mathrm{i} \gamma_1} - e^{-\mathrm{i} \gamma_2}}{e^{-\mathrm{i} \gamma_2} - e^{\mathrm{i} \gamma_1}} \right)^2 \sin^2\left(\frac{\gamma_1+\gamma_2}{2}\right) e^{-2 \xi_1 + 3 \mathrm{i} \gamma_1} \right].
\end{align*}
With the same computations, this yields the following explicit expression
\begin{align}
\lim_{\epsilon \to 0} g &=
-4 \mathrm{i} \alpha^2 e^{- \xi - \mathrm{i} \eta + \frac{\mathrm{i} \gamma}{2}}
\left[ -\cot \gamma + \beta (x - \tilde{x}_0) - \mathrm{i} \alpha (t - \tilde{t}_0) - \frac{\mathrm{i}}{2} \right. \notag \\
& \qquad \left. - e^{-2 \xi + \mathrm{i} \gamma}
\left( \cot \gamma +
\beta (x - \tilde{x}_0) + \mathrm{i} \alpha (t - \tilde{t}_0) - \frac{\mathrm{i}}{2} \right) \right],
\label{g-1}
\end{align}
The exponential double-solitons are given by the explicit expression (\ref{exp-double-sol}) with $f$, $h$, and $g$ given by (\ref{f-3}), (\ref{h-1}), and (\ref{g-1}). By using translational symmetry, we can redefine
$$
\tilde{x}_0 \to \tilde{x}_0 - \frac{1}{\sin \gamma}, \quad
\tilde{t}_0 \to \tilde{t}_0 - \frac{1}{2 \sin \gamma}
$$
to obtain exactly the same expressions as in Theorem \ref{explicit-expression-NU-Muv} for $f = D(M)$, $h = N_v$, and $g = \bar{N}_u$.

\section*{Acknowledgments}

Z.-Q. Li was supported   by the Graduate  Innovation  Program  of China University of Mining and Technology under Grant No. 2024WLKXJ117 and the  Postgraduate  Research  \&  Practice  Innovation Program of Jiangsu Province under Grant No. KYCX24\_2676. S.-F. Tian was supported   by the National NaturalScience Foundation of China under Grant No. 12371255 and the Fundamental Research Funds for the Central Universities of CUMT under Grant No. 2024ZDPYJQ1003. D. E. Pelinovsky is supported by the NSERC Discovery grant.

%

\section*{Data availability statement}

The data which supports the findings of this study is available within the article.

\bibliographystyle{plain}

\end{document}